%% file: main.tex
\documentclass[11pt]{article}
\def\confversion{0}

\usepackage{ifthen}

\newcommand{\ignore}[1]{}

\ifthenelse{\equal{\confversion}{1}}
{
	\newcommand{\conf}[1]{#1}
}
{
	\newcommand{\conf}[1]{\ignore{#1}}
}
\ifthenelse{\equal{\confversion}{0}}
{
	\newcommand{\full}[1]{#1}
}
{
	\newcommand{\full}[1]{\ignore{#1}}
}

\input{preamble}

\full{
\title{
	Graph Connectivity and Single Element Recovery \\via Linear and OR Queries}
\author{Sepehr Assadi\footnote{Department of Computer Science, Rutgers University.}  \and
Deeparnab Chakrabarty\footnote{Department of Computer Science, Dartmouth College. Supported in part by the National Science Foundation grant CCF-1813053.} \and
Sanjeev Khanna\footnote{Department of Computer and Information Science, University of Pennsylvania. Supported in part by the National Science Foundation grants CCF-1617851 and CCF-1763514.} 
}
\date{}
}

\conf{
	
	\title{	Graph Connectivity and Single Element Recovery \\via Linear and OR Queries}

	\titlerunning{Graph Connectivity and Single Element Recovery via Linear and OR Queries}
	
	\author{Sepehr Assadi}{Rutgers University, New Brunswick, NJ 08901, USA}{sepehr.assadi@rutgers.edu}{}{Supported in part by NSF CAREER award CCF-2047061, and gift from Google Research}
	
	\author{Deeparnab Chakrabarty}{Dartmouth College, Hanover NH 03755, USA}{deeparnab@dartmouth.edu}{}{Supported in part by the NSF awards CCF-1813053 and CCF-2041920.}
	
	\author{Sanjeev Khanna}{University of Pennsylvania, Philadelphia PA 19104, USA}{sanjeev@cis.upenn.edu}{}{Supported in part by the NSF awards CCF-1763514, CCF-1617851, CCF-1934876, and CCF-2008305.}

	\Copyright{Sepehr Assadi, Deeparnab Chakrabarty, Sanjeev Khanna} 
	

	\authorrunning{S. Assadi, D. Chakrabarty, S. Khanna}
	
	\EventEditors{Petra Mutzel, Rasmus Pagh, and Grzegorz Herman}
	\EventNoEds{3}
	\EventLongTitle{29th Annual European Symposium on Algorithms (ESA 2021)}
	\EventShortTitle{ESA 2021}
	\EventAcronym{ESA}
	\EventYear{2021}
	\EventDate{September 6--8, 2021}
	\EventLocation{Lisbon, Portugal}
	\EventLogo{}
	\SeriesVolume{204}
	\ArticleNo{5}
	
	\relatedversion{A full version of this paper is available at \href{https://arxiv.org/abs/2007.06098}{https://arxiv.org/abs/2007.06098}. }
}

\begin{document}
\maketitle

\full{\pagenumbering{roman}}

\input{abstract}

\full{
\clearpage
\setcounter{tocdepth}{2} 
\tableofcontents
\clearpage
\pagenumbering{arabic}
\setcounter{page}{1}
}

\clearpage
\setcounter{page}{1}
\input{intro}

\input{lb-single-element}

\conf{
	\input{warmup}

	\input{icalp-det-graph-conn}}
\full{
\input{lb-graph-conn-or}
\input{warmup}
\input{det-graph-conn}

\input{rand-graph-conn}
\input{related}
}

\full{\subsection*{Acknowledgements} We thank anonymous reviewers whose detailed comments have improved the paper.}

\full{
\vspace{3ex}
\bibliographystyle{abbrv}
\bibliography{general}
}
\conf{
\newpage
\bibliographystyle{plainurl}
\bibliography{general}
}

\noindent
\clearpage
\full{
\appendix
\input{appendix}
}
\conf{
}



\end{document}

%% file: preamble.tex
\usepackage{graphicx} 
\usepackage{array} 
\full{\usepackage{amsthm,amsmath, amssymb, amsfonts, verbatim}}

\full{\usepackage{hyphenat,epsfig,multirow}
\usepackage[usenames,dvipsnames]{xcolor}
\usepackage[ruled]{algorithm2e}

\usepackage[utf8]{inputenc}
}
\usepackage{tikz}
\usetikzlibrary{arrows}
\usetikzlibrary{arrows.meta}
\usetikzlibrary{shapes}
\usetikzlibrary{backgrounds}
\usetikzlibrary{positioning}
\usetikzlibrary{decorations.markings}
\usetikzlibrary{patterns}
\usetikzlibrary{calc}
\usetikzlibrary{fit}
\usetikzlibrary{snakes}
\usetikzlibrary{shadows.blur}
\usetikzlibrary{petri,decorations.markings}

\usepackage{caption}
\usepackage{subcaption}

\usepackage{tcolorbox}
\tcbuselibrary{skins,breakable}
\tcbset{enhanced jigsaw}

\full{\usepackage[compact]{titlesec}}

\definecolor{DarkRed}{rgb}{0.5,0.1,0.1}
\definecolor{DarkBlue}{rgb}{0.1,0.1,0.5}

\usepackage{nameref}
\definecolor{ForestGreen}{rgb}{0.1333,0.5451,0.1333}
\definecolor{Red}{rgb}{0.9,0,0}
\full{
	\usepackage[linktocpage=true,
	pagebackref=true,colorlinks,
	linkcolor=DarkBlue,citecolor=ForestGreen,
	bookmarks,bookmarksopen,bookmarksnumbered]
	{hyperref}
\usepackage[noabbrev,nameinlink]{cleveref}
}

\usepackage{bm}
\usepackage{url}
\usepackage{xspace}
\full{\usepackage[mathscr]{euscript}}

\usepackage{mdframed}

\usepackage[noend]{algpseudocode}
\makeatletter
\def\BState{\State\hskip-\ALG@thistlm}
\makeatother

\usepackage{cite}

\usepackage{enumitem}

\full{
\usepackage[margin=1in]{geometry}}
\usepackage{thmtools} 
\usepackage{thm-restate}

\full{
\newtheorem{theorem}{Theorem}
\newtheorem{lemma}{Lemma}[section]

\newtheorem{corollary}[theorem]{Corollary}
\newtheorem{claim}[lemma]{Claim}

\newtheorem{definition}{Definition}[section]

\theoremstyle{definition}

\newtheorem{remark}[lemma]{Remark}

\newtheorem{invariant}[lemma]{Invariant}

\newtheorem*{claim*}{Claim}
\newtheorem*{proposition*}{Proposition}
\newtheorem*{lemma*}{Lemma}
\newtheorem*{problem*}{Problem}
}
\newtheorem{fact}[lemma]{Fact}
\newtheorem{problem}{Problem}

\newtheorem{mdresult}[theorem]{Result}
\newenvironment{Theorem}{\begin{mdframed}[backgroundcolor=lightgray!40,topline=false,rightline=false,leftline=false,bottomline=false,innertopmargin=2pt]\begin{mdresult}}{\end{mdresult}\end{mdframed}}

\full{
\allowdisplaybreaks
}

\renewcommand{\qed}{\nobreak \ifvmode \relax \else
      \ifdim\lastskip<1.5em \hskip-\lastskip
      \hskip1.5em plus0em minus0.5em \fi \nobreak
      \vrule height0.6em width0.6em depth0.0em\fi}

\newcommand{\Qed}[1]{\ensuremath{\qed_{\textnormal{~#1}}}}

\newcommand{\Ot}{\ensuremath{\widetilde{O}}}
\newcommand{\eps}{\ensuremath{\varepsilon}}

\newcommand{\paren}[1]{\ensuremath{\left(#1\right)}\xspace}
\newcommand{\card}[1]{\left\vert{#1}\right\vert}
\newcommand{\Omgt}{\ensuremath{\widetilde{\Omega}}}

\newcommand{\IR}{\ensuremath{\mathbb{R}}}

\newcommand{\IS}{\ensuremath{\mathbb{S}}}

\newcommand{\ceil}[1]{{\left\lceil{#1}\right\rceil}}

\newcommand{\set}[1]{\ensuremath{\left\{ #1 \right\}}}
\newcommand{\poly}{\mbox{\rm poly}}
\newcommand{\polylog}{\mbox{\rm  polylog}}

\newcommand{\ALG}{\ensuremath{\mbox{\sc alg}}\xspace}

\DeclareMathOperator*{\Prob}{\ensuremath{\mathbb{P}}}
\renewcommand{\Pr}{\Prob}

\newcommand{\etal}{{\it et al.\,}}

\newenvironment{tbox}{\begin{tcolorbox}[
		enlarge top by=5pt,
		enlarge bottom by=5pt,
		 breakable,
		 boxsep=0pt,
                  left=4pt,
                  right=4pt,
                  top=10pt,
                  arc=0pt,
                  boxrule=1pt,toprule=1pt,
                  colback=white
                  ]
	}
{\end{tcolorbox}}

\newcommand{\supp}{\ensuremath{\textnormal{supp}}\xspace}

\newcommand{\QQ}{\mathbb{R}}
\newcommand{\calC}{{\cal C}}
\newcommand{\ProblemOne}{{Single Element Recovery}\xspace}
\newcommand{\ALGO}{\mathscr{A}}

\newcommand{\ba}{\boldsymbol{a}}

\newcommand{\bv}{\boldsymbol{v}}
\newcommand{\bu}{\boldsymbol{u}}
\newcommand{\by}{\boldsymbol{y}}

\newcommand{\cQ}{\mathcal{Q}}
\newcommand{\calP}{{\cal P}}
\newcommand{\calE}{{\cal E}}
\newcommand{\calF}{{\cal F}}
\newcommand{\RR}{\mathbb{R}}
\newcommand{\BIS}{\textsf{BIS}\xspace}
\renewcommand{\IS}{\textsf{IS}\xspace}
\newcommand{\ADD}{\textsf{Additive}\xspace}

\newcommand{\UR}{\mathbf{UR}}

\newcommand{\bone}{\mathbf{1}}
\newcommand{\linear}{\textsf{Linear}\xspace}
\newcommand{\cross}{\textsf{Cross}\xspace}
\newcommand{\OR}{\textsf{OR}\xspace}

\newcommand{\MultiCutSampler}{\ensuremath{\textnormal{\textsf{Partition-Sampler}}}\xspace}
\newcommand{\DetBinSearch}{\mathsf{BinarySearch}\xspace}
\newcommand{\DetFindMany}{\mathsf{DetFindMany}\xspace}
\newcommand{\DetFindEdge}{\mathsf{DetFindEdge}\xspace}
\newcommand{\DetGraphConn}{\mathsf{DetGraphConn}\xspace}
\newcommand{\RandGraphConnOR}{\mathsf{RandGraphConnOR}\xspace}
\newcommand{\RandGraphConnBIS}{\mathsf{RandGraphConnBIS}\xspace}
\newcommand{\RandGraphConnCross}{\mathsf{RandGraphConnCross}\xspace}

\newcommand{\BndSuppRec}{\mathsf{BndSuppRec}\xspace}
\newcommand{\LSample}{\mathsf{RandSuppSamp}\xspace}
\newcommand{\SuppEst}{\mathsf{SuppEst}\xspace}
\newcommand{\RandSuppRec}{\mathsf{RandSuppRec}\xspace}
\newcommand{\DegEst}{\mathsf{DegEst}\xspace}
\newcommand{\FindNbrs}{\mathsf{FindNbrs}\xspace}
\newcommand{\RandEdges}{\mathsf{RandEdges}\xspace}

\newcommand{\adv}{\ensuremath{\textsc{adversary}}\xspace}

\newcommand{\Qk}{\ensuremath{Q^{(k)}}}

%% file: abstract.tex
\begin{abstract}
	\noindent
	We study the problem of finding a spanning forest in an undirected, $n$-vertex multi-graph  under two  basic query models. One are \linear queries which are linear measurements on the incidence vector induced by the edges; the other are the weaker \OR queries which only reveal whether a given subset of plausible edges is empty or not. 
	At the heart of our study lies a fundamental problem which we call the {\em single element recovery} problem: given a non-negative vector $x \in \IR^{N}_{\ge 0}$, the objective is to return a single element $x_j > 0$ from the support. 
	Queries can be made in rounds, and our goals is to understand the trade-offs between the query complexity and the rounds of adaptivity needed to solve these problems, for both deterministic and randomized algorithms.
	These questions have connections and ramifications to multiple areas such as sketching, streaming, graph reconstruction, and compressed sensing. \medskip

\noindent
	Our main results are as follows:
	\begin{itemize}
			\item For the single element recovery problem, it is easy to obtain a deterministic, $r$-round algorithm which makes $(N^{1/r}-1)$-queries per-round. We prove that this is tight: any $r$-round {\em deterministic} algorithm must make $\geq (N^{1/r} - 1)$ \linear queries in some round.	
			In contrast, a $1$-round $O(\polylog{(N)})$-query {\em randomized} algorithm is known to exist.
			
		\item We design a {\em deterministic} $O(r)$-round, $\tilde{O}(n^{1+1/r})$-\OR query algorithm for graph connectivity.
		We complement this with an $\tilde{\Omega}(n^{1 + 1/r})$-lower bound for any $r$-round deterministic algorithm in the \OR-model.
	
		\item We design a {\em randomized}, $2$-round algorithm for the graph connectivity problem which makes $\Ot(n)$-\OR queries. In contrast, we prove that any $1$-round algorithm (possibly randomized) requires $\Omgt(n^2)$-\OR queries.	
		A {\em randomized}, $1$-round algorithm making $\Ot(n)$-\linear queries is already known.
	\end{itemize}
\noindent
All our algorithms, in fact, work with more natural graph query models which are {\em special} cases of the above, and have been extensively studied in the literature. These are \cross queries (cut-queries) and \BIS (bipartite independent set) queries.
In particular, we show a {\em randomized}, $1$-round (non-adaptive) algorithm for the graph connectivity problem which makes only $\Ot(n)$-\cross queries. 
	

%
%
%
	
\end{abstract}

%% file: intro.tex

\section{Introduction}\label{sec:intro}

Many modern applications compel algorithm designers to rethink random access to input data, and revisit basic questions
in a {\em query access model} where the input is accessed only via answers to certain kinds of queries. 
There are many reasons for this ranging from data volume (only snapshots of the data can be accessed) to data ownership (access is restricted via certain APIs).

In this paper, we study algorithms accessing an unknown, undirected multi-graph $G$ on $n$ vertices in the following two basic query models.
Think of the graph as an unknown non-negative $\binom{n}{2}$ dimension vector $x_G$ with $\supp(x_G)$ denoting the positive coordinates. With this view, answers to these queries below can be interpreted as {\em measurements} on this vector.
\begin{itemize}[leftmargin=15pt,noitemsep,label=$-$]
	\item \textbf{Linear Queries} (\linear): Given any non-negative\footnote{Non-negativity is for convenience. A general linear query can be broken into two non-negative queries.} $\binom{n}{2}$ dimension vector $a_G$, what is $a_G\cdot x_G$?
	

	\item \textbf{OR Queries} (\OR): Given any subset $S$ of the $\binom{n}{2}$ dimensions, is $\supp(x_G)\cap S$ empty?
	
\end{itemize}

\noindent
Reverting back to the combinatorial nature of graphs, it is perhaps more natural to think of different kinds of queries, and indeed
the following two have been extensively studied. These are however special\footnote{%
	The \cross (and \BIS queries) correspond to $\{0,1\}$ vectors $a_G$ (and subsets) corresponding to cuts.
	Indeed, our algorithms work with the weaker queries while our lower bounds will be for the stronger queries.
	It should also be clear that the \linear (and respectively \cross) queries are at least as strong as \OR (resp, \BIS) queries.
} cases, respectively, of the queries above.
\begin{itemize}[leftmargin=15pt,noitemsep,label=$-$]
	\item \textbf{Cross-additive Queries} (\cross): 
	Given two disjoint subsets $A,B$ of $V$,  $\cross(A,B)$ returns the number of edges, including multipicity,
	that have one endpoint in $A$ and the other in $B$.
	
	\item \textbf{Bipartite Independent Set Queries} (\BIS): 
		Given two disjoint subsets $A,B$ of $V$,  $\BIS(A,B)$ returns 
		whether or not there is an edge
		that has one endpoint in $A$ and the other in $B$.
		
\end{itemize}

\noindent
%
The above query models (and similar variants such as additive queries~\cite{GrebinskiK98}, cut-queries~\cite{RubinsteinSW18}, edge-detection queries~\cite{AngluinC08,BeameHRRS18}) have a rich literature~\cite{GrebinskiK98,AlonBKRS02,ChoiK08,Mazzawi10,Choi13,BeameHRRS18,RubinsteinSW18,AbasiB18,Nis19}. 	Most previous works, however, have focused on either {\em graph reconstruction}~\cite{ChoiK08,Choi13,BshoutyM12,Mazzawi10}, or on {\em parameter estimation} (e.g., estimating the number
of edges~\cite{BeameHRRS18} or triangles~\cite{BhattacharyaBGM18}). 
In this work, however, our goal is to understand the power and limitations of these queries to reveal {\bf \em structural properties} of the underlying graph. In particular, we study the following basic property.

\begin{problem}[Graph connectivity]\label{prob:connectivity}
Given query access to an undirected multigraph 
on the vertex set $V = \set{1,\ldots,n}$, return a {\em spanning forest}. 
%
\end{problem}

\noindent
It is not too hard to implement the {\em classic} BFS or DFS traversals 
to obtain an $\tilde{O}(n)$-query deterministic algorithm for the above problem in either query model. However, such algorithms are {\em adaptive}, that is, the queries depend on the answers obtained so far. 
A much more modern algorithm of Ahn, Guha, and McGregor~\cite{AhnGM12} gives\footnote{Using results in~\cite{SunWYZ19}, one can also obtain a $\tilde{O}(n)$-query deterministic algorithm in the \cross-query model.} an $\tilde{O}(n)$-\linear query {\em non-adaptive} but {\em randomized} algorithm for the problem.
This raises the following questions that motivate us

\begin{quote}
	\emph{
	What is the rounds-of-adaptivity versus query-complexity trade-off for deterministic algorithms for \Cref{prob:connectivity}? Can randomization also help in the \OR and \BIS models? 
	}
\end{quote}
It turns out that understanding the complexity of~\Cref{prob:connectivity} is closely related to understanding an even more basic problem which we discuss below.

%

\subsubsection*{Single Element Recovery.} 
Consider a non-negative real-valued vector $x \in \IR_{\geq 0}^{N}$ and
%
suppose we have access to $x$ only via \linear or \OR queries where the dimension is now $N$.
We define the following problem  which
we call the \emph{single-element recovery} problem (following the standard ``support-recovery'' problem in compressed sensing).  

\begin{problem}[Single-element recovery]\label{prob:vector} 
	Given a non-negative real-valued vector $x \in \IR_{\geq 0}^N$, accessed via either \linear-queries or \OR-queries, 
	output {\em any arbitrary} element\footnote{In the case of \OR-queries, we can only return the $j$ with $x_j > 0$} from the support $\supp(x)$. 
\end{problem}
\noindent
To see how the above problem relates to~\Cref{prob:connectivity}, consider the vector of possible edges incident to a single vertex. A spanning forest must find an edge incident to this vertex. This corresponds to solving~\Cref{prob:vector} on this vector.
%
%
The problem is also interesting in its own right, with connections to {\em combinatorial group testing}~\cite{Dorfman43,DH00,NgoD00}, {\em compressed sensing}~\cite{Donoho06,CormodeM06,Indyk08}, and {\em coin-weighing problems.}~\cite{Lindstrom71,RV97,GrebinskiK98,Bshouty09}. While most of these works have focused on recovering the full support, we ask the simpler question of just recovering a single element. 

If one allows {\em randomization}, then one can use $\ell_0$-samplers~\cite{JowhariST11} to solve the above problem using $O(\log^2 N \log (\frac{1}{\delta}))$ \linear queries\footnote{A similar result holds with \OR queries as well. See~\Cref{sec:warmup}}, {\em non-adaptively}. In fact, $\ell_0$-samplers return a random element in $\supp(x)$.
The parameter $\delta$ is the error probability. There have been numerous applications of these (see the table in Figure 1 of~\cite{KapralovNPWWY17}, for instance), and indeed many applications (including the AGM~\cite{AhnGM12} algorithm alluded to above) need only an arbitrary element in the support. This is precisely what is asked in~\Cref{prob:vector}.
Furthermore, the upper bound for randomized algorithms is nearly tight~\cite{JowhariST11,KapralovNPWWY17}, and therefore, for randomized algorithms, our understanding is pretty much complete.
But what can be said about {\em deterministically} finding a single support element\footnote{A ``deterministic $\ell_0$ sampler'', if you allow us the abuse of notation.}? This is an important question for it relates to deterministic analogs to the various applications stated above.

It is not too hard to make a couple of observations. One, any {\em non-adaptive} deterministic algorithm for~\Cref{prob:vector} using \linear-queries can in fact be recursively used to completely recover the whole vector.
This implies an $\tilde{\Omega}(N)$ information theoretic lower bound. Two, if one allows more rounds, then one can indeed do better using a binary-search style idea. More precisely, in each round the algorithm partitions the search space into $N^{1/r}$ parts and using $N^{1/r}$ queries finds a non-zero part. In this way
in $r$ rounds, one can gets algorithm making $N^{1/r}$-queries per round.
This leads to the following fundamental question which we answer in our paper.

\begin{quote}
	\emph{What is the rounds-of-adaptivity versus query-complexity trade-off for deterministic algorithms for~\Cref{prob:vector}?}
\end{quote}

\subsection{Motivation and Perspective} 
Why should we care about the questions above? 
\begin{itemize}
	\item We think that algorithmic question of computation on graphs via queries is as natural and important as the reconstruction question.
	Indeed, our study was inspired by trying to understand the power of {\em cut}-queries to check whether a graph was connected or not; this is an (extremely) special case of submodular function minimization.
	More recently, this type of ``property-testing via queries'' question on graphs has been asked for matchings by Nisan~\cite{Nis19}, and more generally for matrix properties by~\cite{SunWYZ19} and~\cite{RWZ20}.	
	Single element recovery is also as natural as whole-vector recovery.
Indeed,	one can imagine a scenario where recovering a big\footnote{As we show later in~\Cref{lem:det-r-round-supp-recovery}, algorithmically we can get results when the ``single'' in single element recovery can be larger.}
subset of the support (diseased blood samples, say) faster and with fewer queries may be more  beneficial than reconstructing the whole vector. 
	
	\item The \linear query model is closely connected to {\em linear sketches} that have found plenty of applications in dynamic streaming; see, e.g.~\cite{Ganguly08,LiNW14,KallaugherK20}. 
	The single element recovery problem also has connections to the \emph{universal relation} $\UR^{\subset}$ problem in communication complexity, which was studied in~\cite{KapralovNPWWY17,NelsonY18}.
	Understanding these questions, therefore, have ramifications to other areas. As a concrete example, one consequence of our results is a deterministic, $O(r)$-pass dynamic streaming algorithm for graph connectivity in $\Ot(n^{1+1/r})$ space.
	This was not known before.

	\item We believe the question of the trade-off between rounds versus query complexity is natural and important, especially in today's world of massively parallel computing.
	Such trade-offs are closely related to similar questions in communication complexity, number of passes in streaming algorithms, etc. It is worthwhile building up an arsenal of tools to attack such questions.
	Indeed, one main contribution of this paper is to show how LP-duality can be used as one such tool.	
	
	\item Why do we focus on deterministic algorithms? Mainly because, as mentioned above, our understanding of the complexity of randomized algorithms for the problems above is near complete.
However, in some applications one may require exponentially low error, or has to deal with an ``adversary'' (say, the one giving updates to a streaming algorithm) that is not oblivious to the algorithm's randomness; see, e.g.~\cite{Ben-EliezerJWY20}. 
This further motivates the study of deterministic algorithms in this context. Furthermore, we need to design lower-bounding techniques which only work against deterministic algorithms, and this is of technical interest.

\end{itemize}

\full{The above questions on round-versus-query-complexity trade-offs for both graph connectivity and single element recovery can be interpreted as asking bounds on {\em deterministic adaptive sketching}~\cite{KamathP19,AhnGM12b}.
The answers to the various queries can be thought of as the sketch. With \linear queries, the question is related to linear sketches, which in turn, is closely related to dynamic streaming~\cite{Ganguly08,LiNW14,KallaugherK20}. 
For instance, as we state below, one consequence of our results is a deterministic, $O(r)$-pass dynamic streaming algorithm for graph connectivity in $\Ot(n^{1+1/r})$ space.
}


\subsection{Our Results}\label{sec:results}
Our first result is a tight lower bound for the question on single element recovery. The binary-search style algorithm mentioned above is the best one can do. 
\begin{Theorem}\label{Thm:lb}
	For the single element recovery with \linear-query access, 
	any $r$-round, deterministic algorithm must make $\geq N^{1/r} - 1$ queries in some round.
\end{Theorem}
We should remind the reader that the above lower bound is for vectors whose domain is non-negative rationals. 
In particular, it does not hold for Boolean vectors\footnote{Indeed, for Boolean vector with \linear queries one can recover the whole vector if the query vector has exponentially large coefficients. Even when the coefficients are small ($\{0,1\}$ even), the vector can be recovered with $O(n/\log n)$-queries which is information theoretically optimal. }.
Moving to the continuous domain allows one to use tools from geometry, in particular duality theory and Caratheodory's theorem, to prove the tight lower bound.
\full{We discuss this in more detail in~\Cref{sec:overview}.}

As mentioned above, \linear queries are stronger than \OR queries, and thus the above lower bound holds for \OR queries as well.
The proof for \OR queries, however, is combinatorial, arguably simpler, and more importantly can be generalized to prove the following lower bound for~\Cref{prob:connectivity} as well.

\begin{Theorem}\label{Thm:det-graph-conn-lb}
	Any $r$-round deterministic algorithm for finding a spanning forest, must make $\Omgt(n^{1 + \frac{1}{r}})$-\OR queries. 
	\full{Formal statement in~\Cref{thm:det-graph-lower}.}
\end{Theorem}

As we explain below, the above smooth trade-off between rounds and query complexity is optimal, even when we allow the weaker \BIS-queries. Algorithmically, we have the following result.
We mention that such a result was not known even using \linear or \cross queries.
A similar lower bound as in \Cref{Thm:det-graph-conn-lb} with \cross-queries is left open.
\begin{Theorem}\label{Thm:det-graph-conn}
	For any positive integer $r$, there exists an $O(r)$-round deterministic algorithm which makes $\tilde{O}\left(n^{1+\frac{1}{r}}\right)$-\BIS queries per round, and returns a spanning forest of the graph.
\full{	Formal statement in~\Cref{thm:det-graph-conn-formal}.} 
\end{Theorem}
It is worth remarking that 
our algorithm with \linear queries (which is implied by the weaker \BIS queries) above also implies an $O(r)$-pass $\tilde{O}(n^{1+1/r})$-space {\em deterministic} algorithm for maintaining a spanning forest in {\em dynamic} graph streams. 
As the edge updates arise, one simply updates the answers to the various queries made in each round. This result was not known before.

Finally, we show that for~\Cref{prob:connectivity}, randomization is helpful in decreasing the number of rounds. 
More precisely, we consider Monte-Carlo algorithms. 

\begin{Theorem}\label{Thm:rand-graph-conn-or}
	There exists a $2$-round randomized algorithm for graph connectivity which makes $\tilde{O}(n)$-\OR queries per round.
	There exists a $4$-round randomized algorithm for graph connectivity which makes $\tilde{O}(n)$-\BIS queries per round.
	Any non-adaptive, randomized algorithm for graph connectivity must make $\Omgt(n^2)$-\OR queries.
\full{	Formal statements in~\Cref{thm:rand-graph-conn-or-formal} and~\Cref{thm:rand-nonadaptive-or-lb-formal}.}
\end{Theorem}
%
%

\noindent
\Cref{tab:summary} summarizes our contributions.
\input{tab-results.tex}

\subsection{Technical Overview}\label{sec:overview}

In this section we give a technical overview of our results. These highlight the main underlying ideas and will assist in reading the detailed proofs which appear in the subsequent sections.\conf{\smallskip}

\full{\paragraph{Overview of~\Cref{Thm:lb}.}}
\conf{\noindent {\bf Overview of~\Cref{Thm:lb}.}}
It is relatively easy to prove an $r$-round lower bound for single element recovery in the \OR-query model via an adversary argument (see~\Cref{sec:lb-sing-el-or}).
At a high level, \OR-queries only mildly interact with each other and can be easily fooled.
Linear queries, on the other hand, strongly interact with each other. To illustrate: if we know $x(A)$ and $x(B)$ for $B\subseteq A$, then we immediately know $x(A\setminus B)$. This is untrue for \OR-queries -- if $x$ has a non-zero entry in both $A$ and $B$, nothing can be inferred about its entries in $A\setminus B$. Indeed, this power manifests itself in the non-adaptive, randomized algorithm using \cross-queries; it is important that we can use subtraction. This makes proving lower bounds against \linear-queries distinctly harder.

In our proof of~\Cref{Thm:lb}, we use duality theory. To highlight our idea, for simplicity, let's consider a warmup {\em non-adaptive} problem. The algorithm has to ask $\ll \sqrt{N}$ queries, and on obtaining the response, needs to return a subset $S\subseteq [N]$ of size $\ll \sqrt{N}$ with the guarantee that $\supp(x)\cap S$ is not empty. Note that if this were possible, then there would be a simple $2$-round $o(\sqrt{N})$-algorithm --- simply query the individual coordinates of $S$ in the second round. This is what we want to disprove. Therefore, given the first round's $\ll \sqrt{N}$ queries, we need to show there exists responses such that {\em no matter} which set $S$ of $\ll \sqrt{N}$ size is picked, there exists a feasible $x\in \RR^N_{\geq 0}$ which sets all entries in $S$ to $0$. Note this is a $\exists\forall\exists$-statement. How does one go ahead establishing this?

We first observe that for a fixed response $\ba$ and a fixed set $S$, whether or not a feasible $x\in \RR^N_{\geq 0}$ exists is asking whether a  system of linear inequalities has a feasible solution. 
Farkas Lemma, or taking the dual,  tells us exactly when this is the case. The nice thing about the dual formulation is that the ``response'' $\ba$ becomes a ``variable'' in the dual program, as it should be since we are trying to find it. To say it another way, taking the dual allows us to assert conditions that the response vector $\ba$ must satisfy, and the goal becomes to hunt for such a vector. How does one do that? Well, the conditions are once again linear inequalities, and we again use duality. In particular, we use Farkas Lemma again to obtain conditions certifying the {\em non-existence} of such an $\ba$. The final step is showing that the existence of this certificate is impossible. This step uses another tool from geometry --- Carathedeory's theorem. Basically, it shows 
that if a certificate exists, then a {\em sparse} certificate must exist. And then a simple counting argument shows the impossibility of sparse certificates. This, of course, is an extremely high-level view and for just the warmup problem. In~\Cref{sec:lb-single-element} we give details of this warmup, an also details of how one proves the general $r$-round lower bound building on it.

The interested reader may be wondering about the two instantiations of duality (isn't the dual of the dual the primal?). We point out that duality can be thought of as transforming a $\exists$ statement into a $\forall$ statement: feasibility is a $\exists$ statement, Farkas implies infeasibility is a different $\exists$ statement, and negating we get the original feasibility as a $\forall$ statement. Since we were trying to assert a $\exists \forall \exists$-statement, the two instantiations of duality hit the two different $\exists$. \conf{\smallskip}

\full{\paragraph{Overview of~\Cref{Thm:det-graph-conn-lb}.} }
\conf{\noindent {\bf Overview of~\Cref{Thm:det-graph-conn-lb}.} }
At a high level, the lower bound for \Cref{prob:connectivity}, the spanning forest problem, boils down to a ``direct sum'' version of \Cref{prob:vector}, the single element recovery problem. Imagine the graph is an $n\times n$ bipartite graph. Therefore, finding a spanning forest requires us finding an edge incident to {\em each} of the $n$ vertices on one side. This is precisely solving $n$-independent versions of \Cref{prob:vector} in parallel. However, note that a single query can ``hit'' different instances at once. The question is, as all direct-sum questions are, does this make the problem $n$-times harder? We do not know the answer for \linear queries and leave this as the main open question of our work. However, we can show that the simpler, combinatorial proof of \Cref{Thm:lb} against \OR-queries does have a direct-sum version, and gives
an almost tight lower bound for \Cref{prob:connectivity}. 
This is possible because \OR-queries, as mentioned in the previous paragraph, have only mild interaction between them.
We show that this interaction cannot help by more than a $\poly(r)$-factor. Our proof is an adversary argument, and a similar argument was used recently by Nisan~\cite{Nis19} to show that matchings cannot be approximated well by deterministic algorithms with \OR-queries. Details of this are given in~\Cref{sec:lb-det-conn-or}. \conf{\smallskip}

\full{\paragraph{Overview of~\Cref{Thm:det-graph-conn}.} }
\conf{\noindent {\bf Overview of~\Cref{Thm:det-graph-conn}.} }
In \Cref{sec:warmup}, we show some simple, folklore, and known results for single element recovery.
We build on these algorithms to obtain our algorithms for \Cref{prob:connectivity}.
With every vertex one associates an unknown vector which is an indicator of its neighborhood.
If one applies the $r$-round binary-search algorithm for the single element recovery problem on each such vector, then in $r$-rounds with $O(n^{1+1/r})$-\BIS queries, for every vertex one can obtain a single edge incident on it. This alone however doesn't immediately help: perhaps, we only detect $n/2$ edges and get $n/2$ disconnected clusters. Recursively proceeding only gives an $O(r\log n)$-round algorithm. And we would like no dependence on $n$.

To make progress, we actually give a more sophisticated algorithm 	for single element recovery than binary search, which gives more and may be of independent interest. In particular, we describe an algorithm (\Cref{lem:det-r-round-supp-recovery}) for single element recovery which in $O(r)$ rounds, and making $N^{1/r}$-queries per round, can in fact return as many as $N^{1/4r}$ elements in the support. 
Once we have this, then for graph connectivity we observe that in $O(r)$ rounds, we get polynomially many edges incident on each vertex.
Thus as rounds go on, the number of effective vertices decreases, which allows us to query more aggressively. Altogether, we get an $O(r)$-round algorithm making only $\tilde{O}(n^{1+1/r})$-\BIS queries. The details of this are described in~\Cref{sec:det-graph-conn}.
\conf{\smallskip}

\full{\paragraph{Overview of~\Cref{Thm:rand-graph-conn-or}.}}
\conf{\noindent {\bf Overview of~\Cref{Thm:rand-graph-conn-or}.}}
In the overview of the deterministic algorithm, we had to be a bit conservative in that even after every vertex found $k$ edges ($k$ being $1$ or $n^{O(1/r)}$) incident on it, we pessimistically assumed that after this step the resulting graph still has $\Theta(n/k)$ disconnected clusters, and we haven't learned {\em anything} about the edges across these clusters. In particular, we allow for the situation that the cross-cluster edges can be dense.
With randomization, however, we get to sample $k$ {\em random} edges incident on a vertex. 
This is where we use the recent result of Holm \etal~\cite{HolmKTZZ19} which shows that if the $k$ incident edges are random, then, as long as $k = \Omega(\log n)$, the number of {\em inter-component} edges between the connected components induced by the sampled edges, is $O(n/k)$. That is the cross-cluster edges are sparse.
Therefore, 
 in a single round
with $\tilde{O}(n)$-randomized \BIS queries, we can obtain a disconnected random subgraph, but one such that, whp, there exist
at most $\tilde{O}(n)$ edges across the disconnected components.

Given the above fact, the algorithm is almost immediate. After round $1$, we are in a sparse graph (where nodes now correspond to subsets of already connected vertices). If we were allowed general $\OR$-queries, then a single round with $\tilde{O}(n)$-\OR queries suffices to learn this sparse graph, which in turn, gives us a spanning forest in the original graph. This follows from algorithms for single element recovery when the vector is promised to be sparse (discussed in \Cref{sec:warmup}).
Unfortunately, these queries may not be \BIS-queries; recall that \BIS-queries are restricted to ask about edges across two subsets.
Nevertheless, we can show how to implement the above idea using $2$-extra rounds with only $\BIS$-queries, giving a $4$-round algorithm.
Details can be found in~\Cref{sec:algo-or}.

To complement the above, we also prove that even with randomization, one cannot get non-adaptive ($1$-round) $o(n^2/\log^2 n)$-query algorithms with $\OR$-queries .
Indeed, the family of examples is formed by two cliques (dense graphs) which could have a single edge, or not, that connects them.
A single collection of $o(n^2/\log^2 n)$-\OR queries cannot distinguish between these two families. Details can be found in~\Cref{sec:lb-rand-one}.
 
%

\subsection{Related Works}
Our work falls in the broad class of algorithm design in the {\em query access model}, where one has limited access to the input.
Over the years there has been a significant amount of work relevant to this paper including in graph reconstruction~\cite{GrebinskiK98,AlonBKRS02,AlonA05,ReyzinS07,ChoiK08,Bshouty09,BshoutyM11,Mazzawi10,BshoutyM12,AngluinC08,AbasiB18}, parameter estimation~\cite{RonT16,DellL18,BeameHRRS18,BhattacharyaBGM18}, minimum cuts~\cite{RubinsteinSW18,AssadiCK19}
sketching and streaming~\cite{FlajoletM85,AMS99,FrahlingIS08,JowhariST11,AhnGM12,KapralovNPWWY17,kapralov2017single,AssadiCK19,NelsonY18,SunWYZ19}, combinatorial group testing, compressed sensing, and coin weighing~\cite{Dorfman43,Bshouty09,CormodeM06,D01,Donoho06,RV97,DH00}.
It is impossible to do complete justice, but in~\Cref{sec:related} we give a little more detailed discussion of some of these works and how they fit in with our paper. 

\full{
\subsection{Notation}
Throughout the paper, for a positive integer $p$, $[p]$ denotes the set $\{1,2,\ldots,p\}$. 
Our randomized algorithms are Monte-Carlo and make a fixed number of queries but fail with some probability.
We use ``with high probability'' or ``whp'' to denote a failure probability of $\frac{1}{\poly(n)}$ where $n$ is the relevant size parameter.
The exponent of the polynomial can be traded off with the constant in the query complexity.
Given an undirected multigraph $G$, and two disjoint subsets $S$ and $T$ of vertices, we use $E(S,T)$ to denote the collection of pairs $(s,t)\in S\times T$ such that there is at least one edge between $s$ and $t$. We use $\partial(S)$ to denote $E(S,S^c)$.
We use  $\tilde{O}(f(n))$ to hide $\polylog(f(n))$-factors.
}

%% file: tab-results.tex
 \def\arraystretch{2}

\setlength{\tabcolsep}{4pt}

\newsavebox{\tabone}

\sbox{\tabone}{
             {\small
        \centering
        \begin{tabular}{|l|c|c|c|c|c|}
  	    \hline
       		\multicolumn{2}{|c|}{}& \multicolumn{2}{c}{{\linear $|$ \cross}} & \multicolumn{2}{|c|}{{\OR $|$ \BIS}}   \\
		\cline{3-6} 
		\multicolumn{2}{|c|}{}& Upper Bound& Lower Bound & Upper Bound & Lower Bound \\
		\cline{1-6} 
             \multirow{2}{50pt}{\textbf{Single Element Recovery}} & Det & $r$ , $N^{1/r}-1$ & $\bm{r}$ , $\bm{N^{1/r}-1}$  & $r$ , $N^{1/r}-1$ & $\bm{r}$ , $\bm{N^{1/r}-1}$ \\	
             \cline{2-6} 
		 & Rand&  $r = 1$ ,  $O(\log^2{N})$ &  $r=1$, $\Omega(\log^{2}{N})$ \cite{JowhariST11} & $r=1$, $O(\log^2{N})$ & $r=1$, $\Omega(\log^2 N)$\\
	\cline{1-6} 
             \multirow{2}{70pt}{\textbf{Graph Connectivity}} & Det & $\bm{O(r)}$ , $\bm{n^{1+1/r}}$ & {\large \bf ?} & $\bm{O(r)}$ , $\bm{n^{1+1/r}}$ & $\bm{r}$ , $\bm{\Omgt(n^{1+1/r})}$ \\	
             \cline{2-6} 
		 & Rand & $\bm{r = 1}$ , $\Ot(n)$~\cite{AhnGM12} $|$ ${\Ot(n)}$~\cite{AhnGM12,SunWYZ19} & $r$, $\Omega(n/\log n)$ & $\bm{r = 2 | 4}$ , $\bm{\Ot(n)}$ & $\bm{r=1}$, $\bm{\Omgt(n^2})$ \\
	   \hline
        \end{tabular}
      }
  }
  
 \begin{table}[ht!]
\begin{tikzpicture}
   \node[fill=white](boz){};
\conf{\node[fill=white, inner xsep=-7pt, inner ysep=0pt](table)[right=-20pt of boz, scale=0.8]{\usebox{\tabone}};}
\full{\node[fill=white, inner xsep=-7pt, inner ysep=0pt](table)[right=-20pt of boz, scale=0.9]{\usebox{\tabone}};}
\end{tikzpicture}
          \caption{\conf{\footnotesize} \emph{Summary of the state-of-the-art and our results for graph connectivity and single-element recovery problems. 
          	In each cell, we write the number of rounds followed by the query complexity per round. All lower bounds are with respect to the stronger model (\linear and \OR).
          	For upper bounds, if there is a discrepancy between the stronger and weaker models, we show this using a $|$ as partition.
          	Bold results are ours. The remaining results are folklore unless a reference is explicitly cited.
          	The {\bf ?} indicates the main open question of our paper.}
        \label{tab:summary}}

    \end{table}


%% file: lb-single-element.tex
\def\bbA{\mathbf{A}}
\def\bbB{\mathbf{B}}
\def\bone{\mathbf{1}}
\def\bzero{\mathbf{0}}
\def\bw{\boldsymbol{w}}
\def\bv{\boldsymbol{v}}
\def\bs{\boldsymbol{s}}
\renewcommand{\epsilon}{\varepsilon}

	
\section{Lower Bound for \ProblemOne}\label{sec:lb-single-element}

In this section,
we prove the following theorem.
\setcounter{theorem}{0}
\begin{theorem}\label{thm:lb-lin-single-element}
	Any $r$-round deterministic algorithm for \ProblemOne must make $\geq (N^{1/r} - 1)$-\linear queries in some round.
\end{theorem}
As discussed in the introduction, this is not difficult to show for \OR-queries (see~\Cref{thm:lb-single-element-or} in ~\Cref{sec:lb-sing-el-or}), however, it takes some work to obtain the result for \linear-queries.
Following the overview in \Cref{sec:overview}, we start by describing the lower bound for a simple ``trapping problem'' problem which illustrates the main ideas. The general proof follows inductively.
One piece of notation before we begin: given any subset $S\subseteq [N]$, we use $\bone_S$ to denote the $N$-dimensional indicator vector of the subset $S$ with $1$ in the index corresponding to elements in $S$.

\subsection{Warmup: A One Round Lower Bound for a Trapping Problem}

In this setting, there are two parameters $k$ and $s$. The former is an upper bound on the number of non-adaptive ($1$-round) queries.
The objective of the algorithm, after obtaining the answers to the queries, is to find a {\em subset} $S\subseteq [N]$ with $|S| \leq s$ such that $x_j > 0$ for {\em some} $j\in S$. 
That is, a subset $S$ which traps an element of the support. Note that if such an algorithm exists, then there is a $2$-round algorithm for single element recovery making $k$ queries in round $1$ and $s$ queries in round $2$.
We also assume that $x$ is a non-zero vector since otherwise $x([N]) = 0$. Furthermore, by scaling, we assume that $x([N]) = 1$. 
The main lower bound statement is the following.
\begin{theorem}\label{thm:toy-lb}
	If $(k+1)s < N$, then there cannot exist such an algorithm.
\end{theorem}
\noindent
Note that if $s$ divides $N$, then $k = \frac{N}{s} - 1$ queries indeed do suffice. So the above theorem is tight.
\begin{proof}
	We let $\bbA$ denote the $k\times N$ matrix corresponding to the $k$ queries arranged as row vectors. 
	We use $\ba\in \RR^k_{\geq 0}$ to denote the answers we will give to fool any algorithm. To find this, fix any subset $S$ with $|S|\leq s$, and
 	consider the following system of inequalities parametrized by the answer vector $\ba$. The only inequalities are the non-negativity constraints.
	\begin{equation}\label{eq:toy-primal}
	\calP(\ba; S) = \{x\in \RR^N_{\geq 0} ~:~ x([N]) = 1 ~~~~ \bbA\cdot x = \ba~~~ x(S) = 0 \} \tag{P}
	\end{equation}
	Note that if $\calP(\ba; S)$ has a feasible solution, then given the answers $\ba$ to its queries, the algorithms {\em cannot} return the subset $S$.
	This is because there is a non-negative $x$ consistent with these answers with $S$ disjoint from its support. In other words, $S$ is {\em safe} for the lower bound w.r.t. $\ba$.
	Therefore, if there exists an answer vector $\ba$ such that {\em every} subset $S\subseteq [N]$ with $|S| \leq s$ is safe with respect to $\ba$, that is $\calP(\ba;S)$ is feasible, then we would have proved our lower bound.
	We use use duality and geometry to prove the existence of this vector (if $(k+1)s < N$).
	
	The first step is to understand when for a {\em fixed} set $S$, the system $\calP(\ba;S)$ is infeasible. This is answered by Farkas Lemma. In particular, consider the following 
	system\footnote{Here $\by$ and $\bone_S$ are  {\em row vectors}. In the general proof, there will be multiple $\by$'s indexed with super-scripts. All of them are row-vectors. Putting an added $^\top$ would be a notational mess.} of inequalities where the variables are the Lagrange multipliers corresponding to the equalities in \eqref{eq:toy-primal}. We note that the variables are {\em free}, since $\calP(\ba;S)$ has only equalities in the constraints.
	For convenience, we have eliminated the variable corresponding to the subset $S$ and have moved it to the right hand side.
	\begin{equation}\label{eq:toy-cs}
	\calC_S := \Big\{~~(y^{(0)}, \by) \in \QQ \times \QQ^{k}: ~~y^{(0)} \cdot \bone_{[N]	} + \by \cdot \bbA \leq \bone_S ~~\Big\} \notag
	\end{equation}
	Farkas Lemma asserts that  $\calP(\ba;S)$ is {\em infeasible} if and only if there exists $\by\in \calC_S$ such that $y^{(0)}\cdot 1 + \by\cdot \ba > 0$. 
	Contrapositively, we get that $\calP(\ba;S)$ is feasible, that is $S$ is safe with respect to $\ba$, iff 
	$y^{(0)} +  \by\cdot \ba \leq 0$ for {\em all} $\by \in \calC_S$. Since we want an answer $\ba$ such that {\em all} subsets $S$ with $|S|\leq s$ are safe, we conclude that such an answer exists if and only if the following system of linear inequalities has a feasible solution.
	\begin{equation}
	\label{eq:toy-dual}
	\cQ :=  \Big\{~~\ba \in \QQ^{k}_{\geq 0}:~~~ \by\cdot \ba \leq - y^{(0)}, ~~~\forall S\subseteq [N], |S|\leq s, ~~\forall (y^{(0)},\by) \in \calC_S ~~\Big\} \tag{D}
	\end{equation}
	
	In summary, to prove the lower bound, it suffices to show that $\cQ$ has a feasible solution, and this solution will correspond to the answers to the queries.
	Suppose, for the sake of contradiction, $\cQ$ is {\em infeasible}. Then, again by Farkas Lemma (but on a different system of inequalities), there exists multipliers $\lambda_t \geq 0$ corresponding to constraints $\left(S_t ~\textrm{s.t.}~ |S_t|\leq s, ~~ (y^{(0)}_t, \by_t) \in \calC_{S_t}\right)$ for some $t = 1\ldots T$ such that
	(P1) $\sum_{t=1}^T \lambda_t \by_t \geq \bzero_{k}$ where $\bzero_{k}$ is the $k$-dimensional all zero (row) vector, and (P2) $\sum_{t=1}^T \lambda_t y^{(0)}_t > 0$. Note that this time $\lambda_t$'s are non-negative since $\cQ$ has inequalities in the constraints.
	
	We can focus on the $\lambda_t$'s which are {\em positive} and discard the rest. The next key observation is to {\em upper bound} the size $T$ of the support. 
	Note that the conditions (P1) and (P2) can be equivalently stated as asserting that the $(k+1)$-dimensional {\em cone} spanned by the vectors $(y^{(0)}_t, \by_t)$ contains a non-negative point with first coordinate positive.
	Caratheodory's theorem (for cones) asserts that any such point can be expressed as a conic combination of at most $(k+1)$ vectors. 
	Therefore, we can assume that $T \leq k+1$.
	
	Now we are almost done. Since $(y^{(0)}_t, \by_t) \in \calC_{S_t}$, we have $y^{(0)}_t \cdot \bone_{[N]} + \by_t \cdot \bbA \leq \bone_{S_t}$.
	Taking $\lambda_t$ combinations and adding, we get (since all $\lambda_t > 0$) that
	\[
	\left(\sum_{t=1}^T \lambda_t y^{(0)}_t\right) \cdot \bone_{[N]} + \left(\sum_{t=1}^T \lambda_t \by_t\right)\cdot \bbA ~\leq ~ \sum_{t=1}^T \lambda_t \bone_{S_t}
	\]
	Since every $|S_t|\leq s$, the support of the right hand side vector is $\leq sT \leq s(k+1)$. The support of the left hand side vector is $\geq N$. This is because the second summation is a non-negative vector by (P1), and the first has full support.
	This contradicts $(k+1)s < N$. Hence, $\cQ$ has a feasible solution, which in turn means there exists answers $\ba$ which foils $\bbA$. This proves \Cref{thm:toy-lb}. 
\end{proof}

\subsection{The General $r$-round Lower Bound}
We begin by formally defining what an $r$-round deterministic algorithm is, and what it means for such an algorithm to successfully solve \ProblemOne.

%

\begin{definition}[$r$-round deterministic algorithm.]
	An $r$-round deterministic algorithm $\ALGO$ 
	proceeds by making a collection of linear queries $\bbA^{(1)} \in \QQ^{k_1 \times N}_{\geq 0}$ and obtains the answer $\ba^{(1)} = \bbA^{(1)} \cdot x$. This is the first round of the algorithm.
	For $1 < i \leq r$, in the $i$th round the algorithm makes a collection of linear queries $\bbA^{(i)} \in \QQ^{k_i\times N}_{\geq 0}$.
	This matrix depends on the history $(\bbA^{(1)}, \ba^{(1)}), \ldots, (\bbA^{(i-1)}, \ba^{(i-1)})$. Upon making this query it obtains the answer $\ba^{(i)} = \bbA^{(i)}\cdot x$. We call $\Pi_r := \left((\bbA^{(1)}, \ba^{(1)}), \ldots, (\bbA^{(r)}, \ba^{(r)})\right)$ the $r$-round {\bf \em transcript} of the algorithm. The output of the deterministic algorithm $\ALGO$ only depends on $\Pi_r$.
	
	A vector $y\in \QQ^N_{\geq 0}$ is said to be {\bf \em consistent} with respect to a transcript $\Pi_r$ if 
	$\bbA^{(i)}\cdot y = \ba^{(i)}$ for all $1\leq i\leq r$. A transcript $\Pi_r = \left((\bbA^{(1)}, \ba^{(1)}), \ldots, (\bbA^{(r)}, \ba^{(r)})\right)$ is {\bf \em feasible} for the algorithm if there is some vector $y$ consistent with respect to it, and if
	the algorithm indeed queries $\bbA^{(i)}$ given the $(i-1)$-round transcript $\left((\bbA^{(1)}, \ba^{(1)}), \ldots, (\bbA^{(i-1)}, \ba^{(i-1)})\right)$
\end{definition}


\begin{definition}
	An $r$-round deterministic algorithm $\ALGO$ is said to successfully solve \ProblemOne if for all non-zero $x\in \QQ^N_{\geq 0}$, 
	upon completion of $r$-rounds the algorithm $\ALGO$ returns a coordinate $j\in [N]$ with $x_j > 0$. 
	In particular, if the algorithm returns a coordinate $j$ given a feasible transcript $\Pi_r$, then {\em every} $x$ that is consistent with $\Pi_r$ must have $x_j > 0$.
\end{definition}

%
For technical reasons, we add a $0$th-round for any $r$-round algorithm. In this round, the query ``matrix'' $\bbA^{(0)}$ is the single $N$-dimensional row with all ones. That is, we ask for the sum of $x_j$ for all $j\in [N]$. We assume that the answer $\ba^{(0)}$ is the scalar $1$ to capture the fact that the vector $x$ is non-zero.

%
%
%
%

%

Next we define the notion of safe subsets with respect to a transcript generated till round $i$. A safe subset of coordinates are those for which there is a consistent vector $x$ whose support is disjoint from the subset, that is, $x_j = 0$ for all $j\in S$, or equivalently $x(S) = 0$ since $x\geq0$.
	\begin{definition}
		Given an $i$-round transcript $\Pi_i = \left((\bbA^{(0)}, \ba^{(0)}), \ldots, (\bbA^{(i)}, \ba^{(i)})\right)$, a subset $S\subseteq [N]$ is {\bf safe} w.r.t. $\Pi_i$ if the following system of linear inequalities
	\begin{equation}\label{eq:primal} \tag{Primal}
	\calP(\ba^{(\leq i)}; S) := \Big\{~~~x \in \QQ^{N} : \begin{cases}
	\bbA^{(j)} \cdot x = \ba^{(j)} & \forall 0\leq j\leq i \\
	x(S) = 0 & \\
	x \geq 0 & 
	\end{cases}~~~\Big\}
	\end{equation}
		has a feasible solution.
	\end{definition}
	\begin{claim}\label{clm:safe}
		If $\Pi_r$ is a feasible $r$-round transcript of an algorithm $\ALGO$ such that {\em all} singletons are safe w.r.t $\Pi_r$, then the algorithm $\ALGO$ cannot be successful in solving \ProblemOne.
	\end{claim}
\begin{proof}
	Given $\Pi_r$, the algorithm $\ALGO$ must return some coordinate $j\in [N]$. However $\{j\}$ is safe. That is, there is a feasible solution $x$ to 
	$\calP(\ba^{(\leq r)}, \{j\})$. Indeed, if $x$ were the input vector, the algorithm would return a coordinate not in the support.
\end{proof}

\begin{definition}[Transcript Creation Procedure]
	Given an $r$-round algorithm $\ALGO$, the transcript creation procedure is the following iterative process. In round $i$,
	given the transcript \\$\Pi_{i-1} := \left((\bbA^{(0)}, \ba^{(0)}), \ldots, (\bbA^{(i-1)}, \ba^{(i-1)})\right)$ upon which the algorithm $\ALGO$ queries 
	$\bbA^{(i)}$, and the transcript creation procedure produces an answer $\ba^{(i)}$ such that $\Pi_i = \Pi_{i-1} \circ (\bbA^{(i)}, \ba^{(i)})$ 
	is feasible.
\end{definition}
\noindent
Our main theorem, which implies~\Cref{thm:lb-lin-single-element}, is the following.

\begin{theorem}[Transcript Creation Theorem]\label{thm:realmain}
	Let $k_1,\ldots, k_r$ and $s_0,s_1,\ldots, s_r$ be positive integers such that 
	$s_0 \leq n-1$ and $(k_i + 1)s_i \le s_{i-1}$ for all $i \geq 1$.
	Then given any $r$-round algorithm $\ALGO$ making $\leq k_i$ queries in round $i$, there is a transcript creation procedure to create an $r$-round transcript 
	such that for all $0\leq i\leq r$, any subset $S\subseteq [N]$ with $|S|\leq s_i$ is safe with respect to $\Pi_i$. 
%
\end{theorem}

\begin{corollary}
	Let $k_1, \ldots, k_r$ be any $r$ positive integers with $\prod_{i=1}^n (k_i + 1) < n$.
	No $r$-round algorithm $\ALGO$ which makes $\leq k_i$ queries in round $i$ can be successful for the \ProblemOne problem.
	In particular, this implies~\Cref{thm:lb-lin-single-element}.
\end{corollary}
\begin{proof}
	Set $s_r = 1$, $s_{r-1} = (k_r + 1)$, and in general, $s_i = (k_r + 1)(k_{r-1} + 1)\cdots (k_{i+1} + 1)$.
	Note that the conditions of~\Cref{thm:realmain} are satisfied. Therefore given any algorithm $\ALGO$ making $\leq k_i$ queries
	in round $i$, we can create a $r$-round transcript such that all singleton sets are safe with respect to $\Pi_r$.
	\Cref{clm:safe} implies $\ALGO$ cannot be succesful.
\end{proof}

\subsubsection{Proof of the Transcript Creation Theorem}
We start with writing the dual representation of safe sets. Fix a subset $S\subseteq [N]$ and a transcript $\Pi_i$.
By Farkas lemma we know that the system $\calP(\ba^{(\leq i)}; S)$ is infeasible only if there exists a infeasibility certificate
\[
\left(\by^{(0)}, \by^{(1)}, \ldots, \by^{(i)}\right) \in \QQ \times \QQ^{k_1} \times \cdots \QQ^{k_i} ~:~ \sum_{j=0}^i \by^{(j)}\cdot \bbA^{(j)}  \leq \bone_S ~~~ \textrm{and} ~~~ \sum_{j=0}^i \by^{(j)}\cdot \ba^{(j)} > 0
\]
Here $\bone_S$ is the $n$-dimensional indicator vector of the subset $S$, that is, it has $1$ in the coordinates $j\in S$ and $0$ otherwise. Taking negations, we get that the system $\calP(\ba^{(\leq i)}; S)$ is {\em feasible}, that is $S\subseteq [N]$ is safe w.r.t $\Pi_{i-1}$,  if and only if the following condition holds
\begin{equation}\label{eq:dual}
	\textrm{$S$ is safe w.r.t. $\Pi_i$ iff}~~~~
\by^{(\leq i)} \cdot \ba^{(\leq i)} := \sum_{j=0}^i {\by^{(j)}} \cdot \ba^{(j)} \leq 0 ~~~~\textrm{for all} ~~~~  \by^{(\leq i)}\in \calC^{(i)}_S \tag{Dual}
\end{equation}
where,
\begin{equation}\label{eq:cs}
\calC^{(i)}_S := \Big\{~~\by^{(\leq i)} := (\by^{(0)}, \by^{(1)}, \ldots, \by^{(i)})\in \QQ  \times \QQ^{k_1} \times \cdots \times \QQ^{k_i}: ~~\sum_{j=0}^i \by^{(j)} \cdot \bbA^{(j)} \leq \bone_S~~\Big\} \notag
\end{equation}
\noindent
We are now ready to prove~\Cref{thm:realmain} via induction on $i$. The above representation is the dual definition of safe sets, and this definition is what is easy to induct with. \smallskip

\noindent
{\em Base Case: $i=0$.} We need to show that any subset $S\subseteq [N]$ of size $|S| \leq s_0 = N-1$ is safe with respect to the transcript $(\bbA^{(0)}, \ba^{(0)})$. To remind the reader, $\bbA^{(0)}$ is just the all ones vector and $\ba^{(0)}$ is just the scalar $1$. Using \eqref{eq:dual}, we need to show for any subset $S\subseteq [N]$ with $|S|\leq N-1$, we must have 
\[
\by^{(0)}\cdot \ba^{(0)} \leq 0 ~~\textrm{for all}~~ \by^{(0)}\in \QQ ~~\textrm{such that}~~ \by^{(0)}\cdot \bbA^{(0)} \leq \bone_S 
\]
However, $\by^{(0)}\cdot \bbA^{(0)}$ is the $n$-dimensional vector which is $\by^{(0)}$ on all coordinates. Since $|S|\leq n-1$, there is some coordinate $j\notin S$ such that $\bone_S[j] = 0$. Thus, $\by^{(0)} \leq 0$ implying $\by^{(0)}\cdot \ba^{(0)} \leq 0$. The base case holds. \smallskip

\noindent
{\em Inductive Case: $i\geq 1$.} Assume the conclusion of the theorem holds for all $0\leq j\leq i-1$.
That is, there is a procedure which has created a transcript $\Pi_{i-1} = \left((\bbA^{(0)}, \ba^{(0)}), \cdots, (\bbA^{(i-1)}, \ba^{(i-1)})\right)$ such that every subset $S\subseteq[N]$ with $|S|\leq s_{i-1}$ is safe w.r.t $\Pi_{i-1}$.
Using \eqref{eq:dual}, we can rewrite this as the following statement
	\begin{equation}\label{eq:ih}
\textrm{
	for all ~~$S\subseteq [N], |S|\leq s_{i-1}$, ~~for all ${\by^{(\leq ~i-1)}} \in \calC^{(i-1)}_S$ ~~we have ~~ $\by^{(\leq ~i-1)}\cdot \ba^{(\leq ~i-1)} \leq 0$.}\tag{IH}
\end{equation}

\noindent
Given $\Pi_{i-1}$, the algorithm $\ALGO$ now queries $\bbA^{(i)}$ in round $i$. Our goal is to find answers $\ba^{(i)}\in \QQ^{k_i}_{\geq 0}$ such that
any subset $S\subseteq [N]$ with $|S|\leq s_i$ is safe w.r.t $\Pi_i = \Pi_{i-1} \circ (\bbA^{(i)}, \ba^{(i)})$.
Again referring to~\eqref{eq:dual},  we need to find $\ba^{(i)}\in \QQ^{k_i}_{\geq 0}$ satisfying the following system of linear inequalities.
\begin{equation}
\label{eq:def-dual-2}
\cQ^{(i)} :=  \Big\{~~\ba^{(i)} \in \QQ^{k_i}_{\geq 0}:~~~ \by^{(i)}\cdot \ba^{(i)} \leq - \left(\by^{(\leq ~i-1)} \cdot \ba^{(\leq ~i-1)} \right), ~~~\forall S\subseteq [N], |S|\leq s_i, ~~\forall \by^{(\leq i)} \in \calC^{(i)}_S ~~\Big\}\notag
\end{equation}
Although it may appear that the above system has infinitely many constraints, it suffices to write the constraints for extreme points for the polyhedra $\calC^{(i)}_S$'s. To complete the proof, we need to show that $\cQ^{(i)}$ is non-empty; if so, we can select any $\ba^{(i)}\in \cQ^{(i)}$ for completing the transcript creation procedure, and proving the theorem by induction.
The next lemma does precisely that; this completes the proof of the theorem. \Qed{\Cref{thm:realmain}}

\begin{lemma}\label{lem:non-empty}
	The system of inequalities $\cQ^{(i)}$ has a feasible solution.
\end{lemma}
\begin{proof}
	For the sake of contradiction, suppose not. Applying Farkas lemma (again), we get the following certificate of infeasibility.
	There exists the tuples
	$(\lambda_t > 0, ~~S_t\subseteq [N]~\textrm{with}~ |S_t|\leq s_i, ~~\by^{(\leq i)}_t\in \calC^{(i)}_{S_t})$ for $1\leq t \leq k_i + 1$ such that
	\begin{enumerate}
		\item[(P1):] $\sum_{t=1}^{k_i + 1} \lambda_t \by^{(i)}_t \geq \bzero_{k_i}$, and 
		\item[(P2):] $\sum_{t=1}^{k_i+1} \lambda_t \left(\by^{(\leq ~i-1)}_t \cdot \ba^{(\leq ~i-1)}\right) > 0$.
	\end{enumerate}
Since $\by^{(\leq i)}_t\in \calC^{(i)}_{S_t}$, we get
$
\sum_{j=0}^i \by^{(j)}_t\cdot \bbA^{(j)} \leq \bone_{S_t}$ for all $1\leq t\leq k_i$. Taking the positive $\lambda_t$-combinations of these inequalities, we get
\begin{equation}\label{eq:contra}
\sum_{t=1}^{k_i + 1} \lambda_t \cdot \left(\sum_{j=0}^i \by^{(j)}_t\cdot \bbA^{(j)}  \right) \leq \sum_{t=1}^{k_i + 1} \lambda_t \bone_{S_t}\tag{P3}
\end{equation}
Now, define $\bw^{(j)} := \sum_{t=1}^{k_i+1} \lambda_t \by^{(j)}_t$ for $0\leq j\leq i$. (P1) above implies (Q1): $\bw^{(i)} \geq \bzero_{k_i}$, and (P2) implies (Q2): $\bw^{(\leq ~i-1)}\cdot \ba^{(\leq ~i-1)} > 0$.
And finally, (P3) translates to 
\begin{equation}\label{eq:contra-2}
\underbrace{\sum_{j=0}^{i-1} \bw^{(j)}\cdot \bbA^{(j)}}_{\textrm{Call this}~ \bu_1} ~~+~~ \underbrace{\bw^{(i)}\cdot \bbA^{(i)}}_{\textrm{Call this}~ \bu_2} \leq \underbrace{\sum_{t=1}^{k_i + 1} \lambda_t \bone_{S_t}}_{\textrm{Call this}~ \bv}\tag{Q3}
\end{equation}
Now we are ready to see the contradiction. First observe that the vector $\bv$ has at most $(k_i+1)s_i$ positive entries since it is the sum of $k_i+1$ vectors each of support $\leq s_i$. Since $\bw^{(i)}$ and $\bbA^{(i)}$ are both non-negative, $\bu_2$ is a non-negative vector. This implies that $\bu_1$ must have $\leq (k_i + 1)s_i$ positive entries. From the conditions of the theorem, we get $(k_i + 1)s_i \leq s_{i-1}$. Thus, $\bu_1$ has $\leq s_{i-1}$ positive entries. This in turn implies there exists a scalar $\theta$ such that 
$\theta \bu_1 \leq \bone_S$ for some subset $S\subseteq [N]$ with $|S|\leq s_{i-1}$. That is, 
\[
\sum_{j=0}^{i-1} (\theta \bw^{(j)})\cdot \bbA^{(j)} \leq \bone_S ~~~\Rightarrow ~~~~ \theta \bw^{(\leq~i-1)} \in \calC^{(i-1)}_S
\]
The induction hypothesis \eqref{eq:ih} implies $\theta \bw^{(\leq~i-1)}\cdot \ba^{(\leq~i-1)} \leq 0$. This contradicts (Q2). This completes the proof of the lemma.
\end{proof}

%% file: warmup.tex
\section{Algorithms Warmup: Algorithms for Single Element Recovery} \label{sec:warmup}
In this section, we state some simple and/or well known algorithms for single element recovery. We will be using these as subroutines for our algorithms for graph connectivity as well.
We state most of our algorithms in the weaker $\OR$-query model remarking what advantage, if any, \linear queries may provide.

\begin{lemma}\label{lem:bin-search}
	Let $x \in \RR^N_{\geq 0}$ be a  non-zero, non-negative vector, and let $r$ be a positive integer. There exists an $r$-round {\em deterministic} algorithm $\DetBinSearch_r(x)$ which makes $(N^{1/r} - 1)$-\OR queries per round, and returns a coordinate $j$ with $x_j > 0$.
	If \linear queries are allowed, then one can recover $x_j$ as well.
\end{lemma}
\begin{proof} (Sketch)
	Divide $[N]$ into $N^{1/r}$ blocks each of size $N^{1-1/r}$, and run \OR query on each block but the last, taking $N^{1/r} - 1$ queries in all.
	If one of them evaluates to $1$, recurse on that for the next $r-1$ rounds. Otherwise, recurse on the last block.
\end{proof}


Next, we state a standard result from the combinatorial group testing and coin-weighing literature~\cite{KautzS64,D01,HwangS87,PoratR08,CormodeM06}
which says that if the support of $x$ is known to be small, then there exist efficient 
{\em one-round} deterministic algorithms to recover the complete vector.


%
%
%
%

\begin{lemma}\label{lem:det-support-checker}\cite{HwangS87,PoratR08}
		Let $x\in \RR^N_{\geq 0}$ be a non-zero vector, and let $d$ be any positive integer.
		There exists a $1$-round (non-adaptive) deterministic algorithm $\BndSuppRec(x,d)$ which makes $O(d^2\log N)$-\OR queries and 
		(a) either asserts $\supp(x) > d$, or (b) recovers the full support of $x$. With \linear queries, the number of queries reduces to $O(d\log N)$.
\end{lemma}
\begin{proof}(Sketch) 
	We give a very high level sketch only for the sake of completeness.
	For the case of $d=1$, take the $O(\ceil{\log N}\times N)$ matrix $A$ where column $i$ is the number $i$ represented in binary.
	Then $Ax$ (the ``\OR product'') points to the unique element in the support.
	To see the {\em existence} of a deterministic procedure for larger $d$, one can proceed by the probabilistic method. If one samples each coordinate with probability $1/d$, then 
	with constant probability the vector restricted to this sample has precisely support $1$ for which the above ``$d=1$'' algorithm can be used to recover it.
	Repeating this $O(d\log N)$ times leads to error probability which swamps the union bound over $\leq N^d$ possible sets, implying the existence of a deterministic scheme. 
	Finally, another $O(d)$ arises since we need to recover all the $\leq d$ coordinates.
	All this can be made explicit by using ideas from error correcting codes; we point the interested reader to~\cite{HwangS87,PoratR08} for the details.
\end{proof}
\noindent
Next we move to randomized algorithms. Here ideas from $F_0$-estimation~\cite{FlajoletM85,AMS99} and $\ell_0$-sampling~\cite{FrahlingIS08,JowhariST11,CormodeD14} give the following algorithms.

\begin{lemma}\label{lem:rand-single-element-recovery}
	Let $x\in \RR^N_{\geq 0}$ be a non-zero vector. There exists a $1$-round (non-adaptive) randomized  algorithm $\LSample(x)$ which makes  $O(\log^2 N\log\left(\frac{1}{\delta}\right))$-\OR queries  and returns a random $j\in \supp(x)$ with probability $\geq 1-\delta$.
\end{lemma}
\begin{proof} (Sketch)
	Suppose we knew the support $\supp(x) = d$. Then, we sample each $j\in [N]$ with probability $1/d$ to get a subset $R\subseteq [N]$.
	With constant probability $\supp(x\cap R) = 1$ and, conditioned on that, it contains a random $j\in \supp(x)$. 
	Therefore, running the algorithm $\BndSuppRec(x\cap R, 1)$ asserted in~\Cref{lem:det-support-checker}, we can find a random $j\in \supp(x)$
	with constant probability. Repeating this $O(\log(1/\delta))$ times gives the desired error probability.
	Since we don't know $\supp(x)$, we run for various powers of $2$ in $1$ to $n$; we are guaranteed success in at least one of the scales.
\end{proof}

\begin{lemma}\label{lem:rand-supp-estimate}~(Theorem 7 in \cite{Bshouty18}, also in~\cite{DamM10,FalJOPS16})
		Let $x\in \RR^N_{\geq 0}$ be a non-zero vector. There exists a $1$-round (non-adaptive) randomized  algorithm $\SuppEst(x)$ which makes $O\left(\log N\cdot \log(1/\delta)\right)$-\OR-queries and returns an estimate $\tilde{s}$ of the support which satisfies $\frac{\supp(x)}{3} \leq \tilde{s} \leq 3\supp(x)$ with probability $\geq 1-\delta$.
\end{lemma}

%% file: lb-graph-conn-or.tex
\section{Lower Bounds for Graph Connectivity with \OR-queries.}

In this section we establish our lower bounds with \OR-queries. In the first subsection, we show that deterministic, $r$-round algorithms for finding a spanning forest in an $n$-vertex undirected graph must make $\Omgt(n^{1+1/r})$-\OR queries. This builds on an adversary style lower bound for single element recovery with \OR-queries. A similar style of argument was used by Nisan~\cite{Nis19} to prove lower bounds for finding approximate matchings using \OR-queries.
It may be instructive to first read this \Cref{thm:lb-single-element-or} in \Cref{sec:lb-sing-el-or}, for the lower bound for the spanning forest problem is a direct-sum version of that.
In the second subsection, we show that non-adaptive ($1$-round) {\em randomized} algorithms for graph connectivity with $\OR$-queries must make $\Omgt(n^2)$-queries. 

\subsection{Lower Bound for Deterministic $r$-Round Algorithms}\label{sec:lb-det-conn-or}
The following theorem formalizes~\Cref{Thm:det-graph-conn-lb}. 

\begin{theorem}\label{thm:det-graph-lower}
	For any integer $r \geq 1$, any $r$-round deterministic algorithm in the $\OR$-query access which returns a spanning forest of any given graph $G(V,E)$ must make $\Omega(\frac{n^{1 + \frac{1}{r}}}{r^2 \cdot \log{n}})$-\OR queries.
\end{theorem}
\begin{proof}
The proof of this result is via an adversary argument. The \adv creates instances on a bipartite graph $G(U,V,E)$ with $n$ vertices on each side. Let $\ALG$ be any deterministic $r$-round algorithm
that is able to find a single edge incident {\em every} vertex $u \in U$; this is a much weaker condition than finding a spanning forest. We show a strategy for \adv that forces $\ALG$ to make more than $t := t(n,r) = \paren{\frac{n^{1+\frac{1}{r}}}{32r^2 \cdot \ln{n}}}$ queries to the graph \emph{in one of its rounds}. This establishes the lower bound in~\Cref{thm:det-graph-lower}. 

For every vertex $u\in U$, let $x_u \in \{0,1,\star\}^n$ denote the incidence vector that the \adv maintains. $x_u(i) = 0$ implies that the pair $(u,i)$ (for $i\in V$) is {\em not} an edge in $G$, $x_u(i) = 1$ implies that $(u,i)$ is an edge in $G$, and $x_u(i) = \star$ implies that the status of $(u,i)$ is still unknown. More precisely, whenever $x_u(i) = \star$, given the \adv responses to the queries made so far, there are two consistent graphs, one containing the edge $(u,i)$ and one not. 
Note that finding a coordinate $x_u(i) = 1$, which is precisely the single element recovery problem, means that $\ALG$ has succeeded in finding an edge incident to vertex $u$.
However, $\ALG$ needs to find an $x_u(i) = 1$ for {\em every} $u$. Thus, the problem is precisely solving $n$ single element recovery problems {\em in parallel}, where the queries are allowed to span multiple vectors corresponding to different vertices.

We start with some definitions. For every vertex $u$, the \adv maintains a set $A_u \subseteq [n]$ of {\bf \em active coordinates}. These are precisely the coordinates $i$ for which $x_u(i) = \star$. Initially, $A_u = [n]$ for all vertices $u$, and all $x_u$'s are all-star.
We say a vertex $u$ is an {\bf \em alive vertex} if $A_u \neq \emptyset$. For alive vertices $u$, the vector $x_u \in \{0,\star\}^n$ has no coordinate set to $1$. Otherwise, we call
$u$ {\em dead}, and for a dead vertex $x_u \in \{0,1\}^n$. That is, the vector corresponding to a dead vertex is completely known.

Every query $Q$ ever made by the algorithm is answered $0$ or $1$ by \adv. We call the former query a $0$-query, and the latter a $1$-query. Note, $Q$ is classified post facto by the answers, and not up front. We call a $1$-query $Q$ {\bf \em explained} if $Q$ contains some $x_u(i)$ which has been already set to $1$. Otherwise, a $1$-query is unexplained. 
On the other hand, for every $0$-query, the \adv maintains that every $x_u(i)$ in $Q$ has been set to $0$. Note that by definition, every explained $1$-query must touch a dead vertex, and every unexplained $1$-query $Q$ must intersect $A_u$ for at least one vertex $u$.
We are now ready to state the goal of the \adv.

\begin{lemma}\label{lem:adv-goal}
	At the end of $r$ rounds, if there exists an alive vertex $u$ such that every {\em unexplained} $1$-query $Q$ is either disjoint from $A_u$ or
	$|A_u \cap Q| \geq 2$, then $\ALG$ cannot return a $j$ in $x_u$ asserting $x_u(j) = 1$.
\end{lemma}
\begin{proof}
	Suppose $\ALG$ does return $j$ asserting $x_u(j) = 1$. Since $u$ was alive, we must have $j\in A_u$ for otherwise $x_u(j) = 0$ since $x_u \in \{0,\star\}^n$.
	Now, consider the vector $\hat{x}_u$ which sets $\hat{x}_u(i) = 0$ wherever $x_u(i) = 0$, $\hat{x}_u(j) = 0$, and $\hat{x}_u(i) = 1$ for all $i\in A_u\setminus j$.
	For all other $u'\neq u$, we set $\hat{x}_{u'}(i) = x_{u'}(i)$ for $i\notin A_{u'}$, and $\hat{x}_{u'}(i) = 1$ for all $i\in A_{u'}$. We claim that $\hat{x}$ is consistent with all the responses. This would show that $\ALG$ fails.
	We only need to argue about unexplained $1$-queries for all other queries are consistent with the $0,1$-coordinates of $x$, and thus, $\hat{x}$. Take such a query $Q$. If $Q\cap A_u =\emptyset$, 
	then $Q\cap A_{u'} \neq \emptyset$ for some $u'$, and all $i\in A_{u'}$ is set to $\hat{x}_{u'}(i) = 1$. Otherwise, $|Q\cap A_u| \geq 2$ implying there is some $j'\neq j$ in $Q\cap A_u$. This has been set to $\hat{x}_{u}(j') = 1$.
\end{proof}
\noindent
To achieve the goal, the \adv maintains the following invariants after every round $k$. 
\begin{enumerate}[noitemsep]
	\item[(I1.)] The number of alive vertices, $a_k$, is at least $n\cdot \left(1 - \frac{k}{2r}\right)$.
	\item[(I2.)] For every alive vertex $u$ and for every unexplained $1$-query $Q$ with $Q\cap A_u \neq \emptyset$, we have $|Q\cap A_u| > n^{1 - \frac{k}{r}}$.
\end{enumerate}
Note that at the beginning, that is after round $k=0$, the invariants do hold. Furthermore, observe that if the invariants hold for $k = r$, then we get the premise of \Cref{lem:adv-goal} and the \adv succeeds in fooling $\ALG$.
All that remains is to show how the adversary {\em answers} the queries made in round $(k+1)$ (for $0\leq k\leq r-1$), and how the various sets are changed so that the invariants are maintained. 

Let $Q_1, Q_2, \ldots, Q_t$ be the queries made by $\ALG$ in round $(k+1)$. We say $Q_\ell$ {\em touches} vertex $u$ if $Q_\ell\cap A_u \neq \emptyset$. Call a query $Q_\ell$ {\em broad} if it touches $> 8r\cdot \ln n$ {\em alive} vertices. Call $Q_\ell$ {\em narrow} otherwise.
The next claim shows that a small number alive vertices can ``take care of'' all broad queries.
\begin{claim}\label{clm:A-large}
	There is a subset $S$ of alive vertices with $|S|\leq \frac{n}{4r}$ such that for
	any broad query $Q_\ell$ there is some $u\in S$ with $Q_\ell\cap A_u \neq \emptyset$.
\end{claim}
\begin{proof}
	Consider the following instance of the set cover problem: we have one set $S(u)$ for every alive vertex, and one element $e(Q_\ell)$ for every broad query $Q_\ell$. Each set $S(u)$ contains element $e(Q_\ell)$ if and only if $Q_\ell\cap A_u \neq \emptyset$.
	By setting a weight of $\frac{1}{8r \cdot \ln{n}}$ on the set $S(u)$ corresponding to each alive vertex $u$, we obtain a fractional set cover for this instance since, by design, every element (broad query) belongs to at least $8r \cdot \ln{n}$ many sets (alive vertices).
	Thus the set cover instance has a fractional set cover of size at most $\frac{n}{8r \cdot \ln{n}}$ as there are $\leq n$ sets. As the integrality gap of set cover LP is $\ln{t} \leq 2\ln{n}$, there is an integral set cover of size $\leq \frac{n}{4r}$. That is, there exists a set $S$ of $\frac{n}{4r}$ alive vertices 
	such that for every broad $Q_\ell$, $Q_\ell\cap A_u \neq \emptyset$ for some $u\in S$.
\end{proof}

The \adv does the following: for every $u\in S$, it sets $x_u(i) = 1$ for all $i\in A_u$ and deems $u$ {\em dead}. By the claim above, this step kills {\bf \em at most $\frac{n}{4r}$} alive vertices.
It responds $1$ to every broad query $Q_\ell$. By the above claim, note that these broad queries are explained $1$-queries.
Furthermore, if there is any narrow query $Q$ with $Q\cap A_u \neq \emptyset$ for $u\in S$, then \adv responds to $1$ to such queries as well, and these are also explained $1$-queries.

Next, the \adv responds to the remaining narrow queries. 
If there exists any alive vertex $u$ which is touched by $ \geq n^{1/r}$ such narrow queries, then \adv sets $x_u(i) = 1$ for all $i\in A_u$, deems it dead, and 
responds $1$ to all narrow queries touching this vertex. These $1$-queries are also explained. Since every narrow query touches at most $8r\cdot \ln n$ alive vertices, and there are
$\leq t = \frac{n^{1+1/r}}{32r^2\ln n}$ narrow queries to begin with, a counting argument shows that there cannot be more than $\frac{n}{4r}$ vertices which touch more than $n^{1/r}$ narrow queries.
Therefore, this step kills {\bf \em at most $\frac{n}{4r}$} alive vertices as well. In the remainder of this $(k+1)$th round, \adv does not kill any more vertices, and so the total number of vertices killed this round is $\leq \frac{n}{2r}$.
Therefore, the number of alive vertices after round $(k+1)$ is
$\geq n\cdot \left(1 - \frac{k}{2r}\right) - \frac{n}{2r} = n\cdot \left(1 - \frac{k+1}{2r}\right)$. Thus, Invariant (I1.) holds after round $(k+1)$.

The only unanswered queries left with are narrow queries such that every remaining alive vertex is touched by $\leq n^{1/r} - 1$ of these queries.
This is like a single instance of the single element recovery problem, and the remainder of this proof is akin to that of \Cref{thm:lb-single-element-or}. For every query $Q_\ell$ and for every alive $u$ with $Q\cap A_u \neq \emptyset$, 
if $|Q\cap A_u| \leq n^{1 - \frac{(k+1)}{r}}$, the \adv sets $x_u(i) = 0$ for all $i\in Q\cap A_u$, and removes these coordinates from $A_u$. Since there are $\leq n^{1/r} - 1$ such queries, the total number of vertices removed from $A_u$ is 
$\leq n^{1 - \frac{(k+1)}{r}} \cdot \left(n^{1/r} - 1\right) = \left(n^{1 - \frac{k}{r}} - n^{1 - \frac{(k+1)}{r}}\right)$.
If, for this particular query, all $x_u(i)$ participating in it is set to $0$, the \adv responds $0$.
Otherwise, it responds $1$. In the latter case, the query is an {\em unexplained} $1$-query. This completes the responses of \adv to all the queries made in this round. We now show that Invariant (I2.) holds.

Fix any alive vertex $u$ which remains alive after round $(k+1)$. Fix any unexplained $1$-query $Q$ (which could also be from a previous round) which intersects $A_u$. 
If $Q$ is from round $(k+1)$, then by the description of the \adv strategy, $|Q\cap A_u| > n^{1-\frac{k+1}{r}}$ for otherwise, the \adv would have removed these coordinates from $A_u$.
If $Q$ is from a previous round,
then since (I2.) held after round $k$, we get that
$|Q\cap A_u| > n^{1-\frac{k}{r}}$ before round $(k+1)$. For every vertex $u$ that remains alive, we know that the \adv removes $\leq \left(n^{1 - \frac{k}{r}} - n^{1 - \frac{(k+1)}{r}}\right)$ vertices from $A_u$.
Therefore, after round $(k+1)$, we still have $|Q\cap A_u| > n^{1-\frac{k}{r}} - \left(n^{1 - \frac{k}{r}} - n^{1 - \frac{(k+1)}{r}}\right) = n^{1 - \frac{k+1}{r}}$. Thus, Invariant (I2.) is maintained.

In sum, this shows how \adv can answer all the queries $Q_1, \ldots, Q_t$ in round $(k+1)$ such that both invariants are maintained. Thus, after round $r$, the \adv can maintain the premise of \Cref{lem:adv-goal}.
This in turn, proves \Cref{thm:det-graph-lower}. 
\end{proof}

 \subsection{$\Omgt(n^2)$-Lower Bound for Randomized Non-adpative Algorithms}\label{sec:lb-rand-one} 
The following theorem is a formalization of the lower bound result stated in~\Cref{Thm:rand-graph-conn-or}.
 \begin{theorem}\label{thm:rand-nonadaptive-or-lb-formal}
 	Any $1$-round (non-adaptive) randomized algorithm 
 	which makes less than $\frac{n^2}{2916\log^2 n}$-\OR queries on a graph, 
 	cannot infer whether the graph is connected or not with probability $\geq \frac{7}{16}$.
 \end{theorem}
 \def\cD{\mathcal{D}}
 The $7/16$ can be made arbitrarily close to $1/2$ by making the constant $2916$ larger; we omit these details.
 To prove the above theorem, by Yao's minimax theorem, it suffices to give a distribution $\cD$ over $n$-vertex graphs such that any {\em deterministic} collection of $\leq \frac{n^2}{2916\log^2 n}$-\OR queries fails on this distribution with probability at least $7/16$.
 To describe $\cD$, we describe how a graph is generated in two steps. In the first step, we assign each vertex in $V$ to either $L$ or $R$ with equal probability. We then insert all possible edges among vertices in $L$, as well as among vertices in $R$ -- that is, the graphs induced by $L$ and $R$ are cliques. Let $G_1(V,E_1)$ denote the graph at this stage. Then with probability $1/2$, we output this graph as the final graph (a No instance), and with probability $1/2$, we sample one of the $|L|\cdot |R|$ edge slots connecting vertices in $L$ to vertices in $R$, uniformly at random, insert this edge (call it $e$), and output the resulting graph as the final graph (a Yes instance). Let $G(V,E)$ denote the final graph.
 
 \begin{lemma}
 	\label{lem:1_round_rand_OR}
 	Let $\cQ$ be any fixed set of $\frac{n^2}{2916 \log^2 n}$ \OR queries. Then with probability at least $7/8$, the answers to queries in $\cQ$ are the same on the graphs $G_1(V, E_1)$ and $G(V, E)$.
 \end{lemma}
 \begin{proof}
 	We partition $\cQ$ into two sets of queries, namely, a set $\cQ_1$ that contains queries of size at least $36 \log^2 n$ (long queries), and a set $\cQ_2$ containing the remaining queries (short queries). We first claim that with probability at least $1 - 1/n$, 
 	the response to {\em all} long queries	is $1$ in $G_1$ (and hence, since $G$ is a supergraph of $G_1$, also in $G$).
	To see this, fix any long \OR-query $Q$ with $|Q| \geq 36\log^2 n$. Let us consider these as possible edges $F$ in the $n$-vertex graph. Note that there must exist a set $S\subseteq V$ of with  $|S| \geq 6 \log n$ such that 
	every edge in $F$ is incident to some vertex of $S$. The reason is that $t$ vertices can contain at most $\binom{t}{2} \leq t^2$ edges.
	Therefore, there exists a set $F' \subseteq F$ with $|F'| \geq 3\log n$ such that $F'$ induces an acyclic subgraph.
	Now, the probability that every edge in $F'$ has one end-point in $L$ and other in $R$, is at most $(1/2)^{3 \log n} = 1/n^3$. 
	This means with probability at least $1 - \frac{1}{n^3}$, at least one edge of $F'$ (and thus $F$) must actually have both endpoints in either $L$ or $R$, implying that edge is in $G_1$.
	Therefore, the response to $Q$ is $1$. Taking union bound over all long queries in $\cQ_1$, we conclude that with probability at least $1 - \frac{1}{n}$, all queries in $\cQ_1$ must be answered $1$ in $G_1$ (and hence $G_2$).
	We refer to this event as ${\cal E}_1$.
 	
 	We now analyze the behavior of short queries. Any query $Q \in \cQ_2$ that evaluates to $1$ on $G_1$, continues to be evaluated so in $G$ since we do not remove any edges in going from $G_1$ to $G$. So it suffices to show that every short query that evaluates to $0$ in $G_1$, also evaluates to $0$ in $G$, whp. We first observe that with probability at least $1 - \frac{1}{n}$, the number of edge slots between $L$ and $R$, that is the quantity $|L|\cdot|R|$, is at least $n^2/9$. This follows from a simple application of Chernoff bounds - each of the sets $L$ and $R$, are of size at least $n/3$ with probability at least $1 - \frac{1}{n}$, implying that $|L||R|$, is at least $n^2/9$ with probability at least $1 - \frac{1}{n}$. We will refer to this event as ${\cal E}_2$.
 	
 	From here on, we condition on the simultaneous realization of both events ${\cal E}_1$ and  ${\cal E}_2$. Now fix a short query $Q \in \cQ_2$ that evaluates to $0$ in $G_1$. Since the Yes-instance chooses an edge slot $e$ among the $|L||R|$ edge slots uniformly at random, the probability that $e$ appears in $Q$, is at most $\frac{36 \log^2 n}{n^2/9} = \frac{324 \log^2 n}{n^2}$. Thus, the expected number of queries in $\cQ_2$ that contains the edge slot $e$ is at most $|\cQ_2|\cdot\frac{324 \log^2 n}{n^2}\leq \frac{1}{9}$. 
 	That is, the probability that some query in $\cQ_2$ contains the edge slot $e$ is at most $1/9$.
 	
 	Putting together, the probability that the set of queries $\cQ$ have different responses on graphs $G_1$ and $G$ is at most $1/9 + 2/n$ which is at most $1/8$ for sufficiently large $n$.
 \end{proof}
 We can now complete the proof of~\Cref{thm:rand-nonadaptive-or-lb-formal} as follows.
 By~\Cref{lem:1_round_rand_OR}, any deterministic algorithm that performs less than $\frac{n^2}{2916 \log^2 n}$ \OR queries sees the same answers on Yes and No instances generated from distribution $\cD$ with probability at least $7/8$. Thus any deterministic algorithm must err with probability at least $7/16$ in distinguishing between the Yes and No instances of $\cD$. As mentioned above, Yao's minimax lemma implies the theorem.

%% file: det-graph-conn.tex
\section{Deterministic Algorithm for Graph Connectivity}\label{sec:det-graph-conn}

In this section, we prove the following theorem which formalizes~\Cref{Thm:det-graph-conn}.
\begin{theorem}\label{thm:det-graph-conn-formal}
	Let $r$ be any fixed positive integer. 
	There exists an $35r$-round deterministic algorithm $\DetGraphConn(G)$ which makes 
	at most $O(n^{1+\frac{1}{r}} \log n)$-\BIS-queries on an undirected multigraph $G$, and returns a spanning forest of $G$.
\end{theorem}

\noindent
We start by establishing some simple subroutines which we need.

\subsection{Simple Subroutines}

We begin by strengthening the simple algorithm $\DetBinSearch$ asserted in~\Cref{lem:bin-search}. While in $r$-rounds with $O(N^{1/r})$-\OR queries $\DetBinSearch_r(x)$ recovers a single element in the support, one can in fact get many more elements from the support. This result may be of independent interest.
%
%
%
%
\begin{lemma}
	\label{lem:det-r-round-supp-recovery}
		Let $x \in \RR^N_{+}$ be a  non-zero, non-negative vector, and let $r$ be a positive integer, and let $c<r$. There exists a $\ceil{2r/c}$-round {\em deterministic} algorithm $\DetFindMany_{r,c}(x)$ which makes $O(N^{c/r}\log N)$-\OR queries per round, and returns $\min(N^{c/4r}, \supp(x))$ distinct coordinates from $\supp(x)$.
	
\end{lemma}
\begin{proof}
	In the first round, we partition the range $[N]$ into $N^{c/2r}$ blocks of size $N^{1-c/2r}$ each. Let these blocks be $B_1, \ldots, B_k$ with 
	$k = N^{c/2r}$.
	For each $i\in [k]$, 
	 we run the algorithm $\BndSuppRec(x\cap B_i, N^{c/4r})$
	asserted in~\Cref{lem:det-support-checker}. The total number of queries used here is $O(N^{c/2r}\cdot \left(N^{c/4r}\right)^2\log N) = O(N^{c/r}\log N)$.

	At the end of this round, either we recover $\supp(x\cap B_i)$ for each block, and thus recover $\supp(x)$, 
	and we are done.
	Or, there is at least one block of size $N^{1- c/2r}$ which is guaranteed to contain $\geq N^{c/4r}$ elements in its support.
	We call this the heavy block of round 1. 
	Next, we now proceed to recover $N^{c/4r}$ elements from this heavy block of round 1. 
	
	In the second round, we partition the indices of this heavy block again into $N^{c/2r}$ blocks of size $N^{1-2c/2r}$ each, and run
	$\BndSuppRec$ again on this block with $d = N^{c/4r}$.
	Once again, either we recover the entire support of the heavy block (which is guaranteed to contain at least $N^{c/4r}$ elements) and we are done. Or find a block of size $N^{1-2c/2r}$ that contains at least $N^{c/4r}$ elements in its support-- this is the heavy block of round $2$ -- and we now proceed to recover $N^{c/4r}$ elements in the heavy block of round 2. 
	
	We continue in this manner, and after $\ceil{2r/c} - 1$ rounds, either we have already recovered at least $N^{c/4r}$ elements in the support of $x$, or have identified a heavy block of size $N^{1 - \left((\frac{2r}{c}-1)\cdot \frac{c}{2r}\right)} = N^{c/2r}$ that contains at least $N^{c/4r}$ elements in the support of $x$. In the final round, we can simply probe each entry completing the proof.
\end{proof}
\begin{remark}
The trade-off between the number of queries and number of elements recovered is not tightly established for the purpose of what we need in the graph connectivity algorithm.
For instance, using the same idea as above, 
in $2$ rounds one can actually recover $\min(N^{1/4},\supp(x))$ coordinates making $O(N^{3/4})$-queries per round.
\end{remark}
Next, we give an algorithm to find edges between two disjoint sets of vertices using $\BIS$-queries.

\begin{lemma}\label{lem:det-set-to-set}
	Let $A$ and $B$ be two {\em disjoint} sets of vertices with at least one edge between them. There exists a $2r$-round deterministic algorithm
	$\DetFindEdge_r(A,B)$ which makes $O(|A|^{1/r} + |B|^{1/r})$-\BIS queries per round, and returns an edge $(a,b)$ with $a\in A$ and $b\in B$.
\end{lemma}
\begin{proof}
	Consider the $|B|$ dimensional vector $x$ where $x_b$ indicates the number of edges from a vertex $b\in B$ to vertices in $A$.
	We can simulate an \OR-query in this vector using a \BIS-query in the graph --- for any subset $S\subseteq B$, $\OR(S)$ on $x$ has the same answer as $\BIS(A,S)$. Therefore, using~\Cref{lem:bin-search}, in $r$-rounds and $|B|^{1/r}$-\BIS-queries, we can find a coordinate $b^*\in B$ with $x_{b^*} > 0$. That is, there is an edge between $b^*$ and some vertex in $A$.
	
	We can find one such vertex $a\in A$ to which $b^*$ has an edge, again as above. We define the $|A|$-dimensional vector $y$ where $y_a$ indicates the number of edges from $b^*$ to $a$. Once again, the \OR-query on $y$ can be simulated using a \BIS query on the graph --- for any subset $S\subseteq A$, $\OR(S)$ on $y$ is the same as $\BIS(S,\{b^*\})$. 
\end{proof}

\subsection{The Connectivity Algorithm}
Now we give the $O(r)$-round deterministic algorithm to find a spanning forest.
First, we need the following simple claim.

\begin{claim}
	\label{claim:conn_components}
	Let $G(V,E)$ be an arbitrary connected multigraph graph on $n$ vertices, and let $D$ be an arbitrary integer in $\{0, 1, \ldots, (n-1)\}$. 
	Let $V_L$ denote all vertices in $V$ whose degree is at most $D$, and let $V_H = V \setminus V_L$.
	Let $E' \subseteq E$ be an arbitrary set of edges that satisfies the following property: for each vertex $u \in V_L$, the set $E'$ contains all edges incident on $u$, and for every each vertex $v \in V_H$, the set $E'$ contains $D$ arbitrary edges incident on $v$. Then the graph $G' = (V,E’)$ contains at most $\lfloor n/D \rfloor$  connected components. 
\end{claim}
\begin{proof}
Suppose $G'$ has $K \geq \frac{n}{D}$ connected components. Thus, there must exist some component $C$ with $\leq D$ vertices.
Firstly, that $C$ must have some vertex $v\in V_H$. If not, then since vertices in $V_L$ have all their edges in $G$ also in $G'$, this component would be disconnected in $G$ which contradicts $G$'s connectedness.
Secondly, observe that this leads to a contradiction: $v$ has degree at least $D$ in $G'$, and since there are at most $D-1$ other vertices in $C$, one of $v$'s neighbor in $G'$ must lie outside $C$.
This contradicts that $C$ is a connected component.
%
\end{proof}

\noindent
We are now ready to describe the algorithm $\DetGraphConn(G)$. For simplicity, assume $G$ is connected and our goal is to find a spanning {\em tree}. Subsequently, we explain how to modify the algorithm to find a spanning forest of a general graph.
The algorithm proceeds in $O(\log r)$ phases starting with phase $0$. 
The input to phase $i$ is a partition 
$\Pi_i = (S_1, \ldots, S_p)$ of the vertices. Each $S_j$ in $\Pi_i$ is guaranteed to be a connected in the graph $G$.
$\Pi_0$ is the trivial partition of $n$ singletons. Given $\Pi_i$, we define the graph $G_i = (\Pi_i,\calE_i)$ where $\calE_i$ is the collection of {\em pseudo-edges}
between components: we have a pseudo-edge $(S_a, S_b)\in \calE_i$ if and only if there exists some edge in $G$ between a vertex $u\in S_a$ and a vertex $v\in S_b$.
Thus, $G_0$ is indeed the original graph. Note that by our assumption that $G$ is connected, all the $G_i$'s are connected.
We will be collecting pseudo-edges which will imply the connected components; we initialize this set $\calF$ to empty set.
We will maintain the following invariant for a phase: $|\Pi_i| \leq n^{1 - \frac{4^{i} - 1}{r}}$; this is certainly true for $i=0$. Next, we describe a phase $i$.
\begin{enumerate}
	\item For each $S\in \Pi_i$, we construct a vector $x$ indexed by all sets in $\Pi_i\setminus S$ where $x_T$ indicates whether there is a pseudo-edge $(S,T)$ in $G_i$. Next, we run the algorithm $\DetFindMany_{r,c}(x)$ asserted in~\Cref{lem:det-r-round-supp-recovery} to 
	either find all pseudo-edges incident on $S$, or at least $n^{4^i/4r}$ of them. To do so, we set $c$ such that $N^{c/4r} = n^{4^i/4r}$,
	where $N$ is the dimension of $x$. That is, $N = |\Pi_i| - 1$.Indeed, we should set $c = \theta\cdot 4^i$ where $N^\theta = n$. Note, $\theta \geq 1$.
	Also note that the $\OR$-queries on $x$ can be simulated using $\BIS$-queries on the original graph $G$. This is because we are looking at edges between $S$ and a union of a subset of parts in $\Pi_i\setminus S$.

	The number of rounds is $\ceil{\frac{2r}{c}} \leq \ceil{\frac{2r}{4^i}}$. The number of $\BIS$-queries per round is $O(N^{c/r}\log N) = O(n^{4^i/r}\log n)$ per subset $S\in \Pi_i$. And thus, the total number of queries made is $N\cdot O(n^{4^i/r} \log n) \leq |\Pi_i|\cdot O(n^{4^i/r} \log n) \leq n^{1 - \frac{4^i - 1}{r}} \cdot O(n^{4^i/r} \log n) = O(n^{1 + 1/r}\cdot \log n)$.
	
	\item Let $\calE'_i \subseteq \calE_i$ be the pseudo-edges obtained from the previous step. Let $\calF'_i$ be an arbitrary spanning forest of $\calE'_i$.
	We add all these edges to the collection $\calF$.
	Note, $\calF'_i$ is a collection of $\leq |\Pi_i|$ pseudo-edges. 
	
	\item 	Applying~\Cref{claim:conn_components} to the graph $G_i$, adding the pseudo edges in $\calE'_i$ reduces the number of connected components to 
	at most $|\Pi_i|/ n^{\frac{4^i}{4r}}$. We now repeat the above two steps $11$ more times {\em sequentially}, and each time the number of connected components multiplicatively drops  by $n^{\frac{4^i}{4r}}$. Thus, after the $12$ sub-phases we end up with the partition $\Pi_{i+1}$ of connected components, with $|\Pi_{i+1}|\leq |\Pi_i| / n^{\frac{12\cdot 4^i}{4r}} \le  n^{1 - \frac{4^{i} - 1}{r}}\cdot n^{-\frac{12\cdot 4^i}{4r}} =  n^{1-\frac{4^{i+1} - 1}{r}}$, as desired. 
	The second inequality follows from the invariant before phase $(i+1)$ started. 
\end{enumerate}

To summarize, Phase $i$ performs $O(\frac{r}{4^i})$-rounds and makes $O(n^{1+1/r}\log n)$-\BIS queries per round.
We run phase $0$ to $L = O(\log r)$, till we get $|\Pi_L|\leq \sqrt{n}$. After than we run a clean up phase.

\begin{enumerate}
	\setcounter{enumi}{3}
	\item {\em Clean-up Phase.} Once $|\Pi_L| = O(\sqrt{n})$, for each pair $(S,T)$ in $\Pi_L \times \Pi_L$, we 
	make a single $\BIS$-query to detect if the pseudo-edge $(S,T)\in \calE_L$. The total number of queries is $O(n)$.
	We add an arbitrary spanning tree of $\calE_L$ to the set $\calF$. At this point, $\calF$ lets us know the structure of connectivity via pseudo-edges. 
	The next step is to recover the actual graph edges.
	
	\item {\em Tree Building Phase.} Note that the total number of pseudo-edges in $\calF$ is $< n-1$. For each $(S,T)\in \calF$, we now desire to find an edge $(s,t)$ in the graph where $s\in S$ and $t\in T$.
	Note that once we do this, we have the spanning tree the graph.
	This can be done in $2r$ more rounds using the algorithm $\DetFindEdge_r(S,T)$ using $O\left(|S|^{1/r} + |T|^{1/r}\right)$-\BIS queries per round.
	Therefore, the total number of queries per round of this phase is $O(n)\cdot O(n^{1/r}) = O(n^{1+1/r})$.
\end{enumerate}
\noindent
The number of rounds is $\sum_{i=1}^{O(\log r)} \frac{24r}{4^i} + 1 + 2r \le 35r$.

This ends the description of the algorithm when $G$ is connected. 
If $G$ had more than one connected component, then
one can recognize the connected components as the algorithm progresses. More precisely, if the algorithm is processing the partition $\Pi_i = (S_1, \ldots, S_p)$ and find that $S_i$ has no edges coming out of it, then 
by the invariant that $S_i$ is connected, the algorithm can discard this component and proceed on the remaining graph as if it were connected. The analysis becomes better as the effective number of vertices decrease but the number of available queries don't.
This completes the proof of~\Cref{thm:det-graph-conn-formal}.

%% file: rand-graph-conn.tex
\section{Randomized Algorithms for Graph Connectivity}
In the first subsection, we give a $2$-round randomized algorithm which makes $\tilde{O}(n)$-\OR queries and returns a spanning forest of $G$.
Recall, \Cref{thm:rand-nonadaptive-or-lb-formal} proved a $\Omgt(n^2)$-lower bound for $1$-round randomized algorithms using $\OR$-queries. We also show how to implement the $2$-round algorithm in $4$-rounds using only $\BIS$-queries.
In the second subsection we give a non-adaptive ($1$-round), randomized algorithm which makes $\tilde{O}(n)$ \cross-queries and returns a spanning forest of $G$. Previously, this result was known to hold only with \linear-queries~\cite{AhnGM12}.

\subsection{\OR-queries}\label{sec:algo-or}
We prove the following theorem. 
\begin{theorem}\label{thm:rand-graph-conn-or-formal}
	There exists a $2$-round randomized algorithm $\RandGraphConnOR(G)$ which makes at most $O(n\log^5 n)$-\OR-queries 
	and returns a spanning forest of $G$ with high probability. 
	There exists a $4$-round randomized algorithm $\RandGraphConnBIS(G)$ which makes at most $O(n\log^5 n)$-\BIS-queries 
	and returns a spanning forest of $G$ with high probability.
\end{theorem}

Before we dive into the algorithm, we first establish some simple subroutines which are implied by the algorithms for single element recovery described in~\Cref{sec:warmup}.

\subsubsection{Simple Subroutines}
We begin with a subroutine for a $1$-round randomized algorithm for support recovery.
\begin{lemma}\label{lem:rand-mult-element-recovery}
	Let $x\in \RR^N_{\geq 0}$ be a non-zero vector and let $M$ be an upper bound on $\supp(x)$. 
	There exists a $1$-round (non-adaptive) randomized  algorithm $\RandSuppRec(x)$ which makes  $O(M \log^3 N)$-\OR queries and recovers $\supp(x)$ whp.
\end{lemma}
\begin{proof}
	This follows by running $O(M\log N)$ copies of the algorithm $\LSample(x)$ from~\Cref{lem:rand-single-element-recovery} with $\delta = 1/2$.
	Whp, we would get $\Theta(M\log N)$ random samples from $\supp(x)$, and by a coupon collector argument, whp these will contain the full support.
\end{proof}
The above lemma, for instance, implies that if a graph has $\leq m$ edges, then there is a $1$-round randomized algorithm to reconstruct the graph using $O(m\log^3 n)$ \OR-queries. Note, these are not necessarily $\BIS$-queries, and we do need little extra work to handle $\BIS$-qeuries. In particular, we need the following lemmas.


\begin{lemma}[A corollary of~\Cref{lem:rand-single-element-recovery}]\label{lem:rand-graph-sample}
	Let $v$ be a vertex, $S$ be a subset of vertices not containing $v$, and $s$ be a positive integer. There exists a $1$-round (non-adaptive) randomized  algorithm $\RandEdges(v,S,s)$ which makes
	$O(s\log^3 n)$ \BIS-queries and, whp, returns $s$ random edges from $v$ to $S$, uniformly at random with repetition.
	If there are no edges from $v$ to $S$, the algorithm says so. We use $\RandEdges(v,s)$ to denote $\RandEdges(v,V\setminus v, s)$.
\end{lemma}

\begin{lemma}[A corollary of~\Cref{lem:rand-supp-estimate}]\label{lem:rand-bis-degree-estimate}
	Let $A$ be a subset of vertices in a graph. There is a $1$-round randomized algorithm $\DegEst(A)$ which makes  $O(\log^2 n)$-\BIS-queries and returns an estimate $d$ of  $|E(A,A^c)|$ which satisfies $\frac{|E(A,A^c)|}{3} \leq d \leq 3|E(A,A^c)|$ whp.
\end{lemma}
\begin{lemma}[A corollary of~\Cref{lem:rand-mult-element-recovery}]\label{lem:rand-graph-cut-vertex-nbrs}
	Let $A$ be a subset of vertices and let $d$ be an upper bound on the number of edges in $E(A,A^c)$.
	There exists a $1$-round randomized algorithm $\FindNbrs(A)$ which makes $O(d \log^3 n)$ \BIS-queries and finds the subset $B\subseteq A^c$ of vertices 
	which has at least one edge to some vertex in $A$, whp. That is, it finds the vertex neighbors of $A$.
\end{lemma}
\begin{proof}
	For the first part, let $x$ be the $|A^c|$ dimensional vector with $x_{b}$ indicating the number of edges between $b\in A^c$ and $A$.
	Note that an $\OR(x,S)$ query on this vector for $S\subseteq A\times A^c$ can be simulated using an $\BIS$-query on the graph as well. Note $\supp(x) \leq d$ (it could be much smaller).
	Using the algorithm asserted in~\Cref{lem:rand-mult-element-recovery}, we can find $\supp(x)$, that is, all $b\in A^c$ which have some neighbor in $A$.
	This is precisely what this lemma asserts.
	
\end{proof}

\subsubsection{The Connectivity Algorithm}
Below we give an algorithm which runs $2$-rounds with $\OR$-queries, and in $4$-rounds with $\BIS$-queries.
\begin{enumerate}
	\item In Round 1, every vertex $v\in V$ whp samples $s = O(\log^2 n)$-edges incident to it using the 
	algorithm $\RandEdges(v,s)$ asserted in~\Cref{lem:rand-graph-sample}.
	This requires  $O(n\log^5 n)$-\BIS queries.
	
	Let $\Pi:= (S_1, \ldots, S_p)$ be the connected components formed by these edges. 
	Due to~\Cref{lem:edge-sampling-sparse}, whp we have that $E_\cross(\Pi) = O(n\log n)$. 
	%
	\item[2a.] If we had $\OR$-queries available, then in Round 2 we next apply the algorithm $\RandSuppRec(x)$ asserted by~\Cref{lem:rand-mult-element-recovery} with $M = O(n\log n)$, on the vector $x$ which indexed by $(u,v)$ for vertex pairs across different components of $\Pi$. This requires $O(n\log^4 n)$ \OR-queries.
	We obtain the support of this vector $x$, that is the set $E_\cross(\Pi)$, whp.
	These edges, along with the edges sampled in Round 1, gives the spanning forest of the graph. 
	This completes the description of $\RandGraphConnOR(G)$.
	
	\item[2b.] If we did not have $\OR$-queries but only $\BIS$-queries, then in Round 2, for each $S_i$, $i\in [p]$, we estimate an upper bound $d_i$ on  
	$|E(S_i,S^c_i)|$ using the algorithm $\DegEst(S_i)$ asserted in~\Cref{lem:rand-bis-degree-estimate}. 
	This round requires $O(n\log^2 n)$ \BIS-queries since $p\leq n$.
	
	\setcounter{enumi}{2}
	\item In Round 3, for each $i\in [p]$, we use the algorithm $\FindNbrs(S_i)$ asserted in~\Cref{lem:rand-graph-cut-vertex-nbrs} to find 
	the vertices $V_i \subseteq S^c_i$ which have at least one edge to at least one node in $S_i$.
	The number of $\BIS$-queries needed is $O(\sum_i d_i \cdot \log^3 n) = O(n\log^4 n)$, since $\sum_{i=1}^p d_i = O(|E_\cross(\Pi)|) = O(n\log n)$.
	
	From the information obtained after Round 3, we can figure out {\em pseudo-edges} $\tilde{E} \subseteq [p]\times [p]$ where $(i,j)\in \tilde{E}$ if
	there is a vertex in $S_j$ which has an edge to $S_i$, or vice-versa. To use the terms defined above, $(i,j)\in \tilde{E}$ if $V_i\cap S_j$ or $V_j\cap S_i$ is non-empty.
	Note that if there is no edge $(i,j)\in \tilde{E}$, there is no edge from any vertex in $S_i$ to any vertex in $S_j$ in the original graph.
	
	Let $\tilde{F}$ be an arbitrary spanning forest in $\tilde{E}$. Note the connected components of $\tilde{F}$ are also the connected components of $G$.
	What remains is to find the edges of $G$ (we only have pseudo-edges now) to connect up the connected components.
	We use the next round to find the actual graph edges.
	
	\item In Round 4, for every pseudo-edge $(i,j)\in \tilde{F}$, we know either $V_i\cap S_j \neq \emptyset$ or $V_j\cap S_i \neq \emptyset$ and we know which is the case.
	Suppose $V_i\cap S_j\neq \emptyset$. Then, let $w_{ij}$ be a vertex in $V_i\cap S_j$. Let $W$ be all such vertices collected as we go over all the pseudo-edges $(i,j)\in \tilde{F}$.	Note $|W|\leq k-1 < n$.
	
	For each vertex $w_{ij} \in W$ in parallel, we use the algorithm $\RandEdges(w_{ij}, S_i, 1)$ asserted in~\Cref{lem:rand-graph-sample} with 	to get a graph edge from $w_{ij}$ to a vertex $u\in S_i$.
	Once we get all such edges, we would have obtained a spanning forest of $G$ whp. The number of queries is $O(|W|\log^3 n)$ which is $O(n\log^3 n)$.
	This completes the description of $\RandGraphConnBIS(G)$.
\end{enumerate}
\noindent
Note that every step in the above description uses randomized subroutines which succeed whp. To prove~\Cref{thm:rand-graph-conn-or-formal} via union bound, we need to ensure that only a polynomially many events occur. Indeed this is the case; the total number of events is $O(n)$. \Cref{lem:edge-sampling-sparse} completes the proof of~\Cref{thm:rand-graph-conn-or-formal}.


\begin{lemma}\label{lem:rand-or-conn-comp}\label{lem:edge-sampling-sparse}
	Let $G$ be an undirected graph on $n$ vertices. Suppose every vertex $v\in G$ samples $O(\log^2 n)$ edges with repetition, and let
	$\Pi = (S_1, \ldots, S_p)$ be the resulting connected components. 
	Then, whp, $|E_\mathsf{cross}(\Pi)| = O(n\log n)$.
\end{lemma}
\begin{proof}
	We use the following strong theorem recently proved by Holm \etal \cite{HolmKTZZ19} in FOCS 2019. 
	The only difference is that the process in~\cite{HolmKTZZ19} is without repetition.
	But it can be easily modified for the with-repetition case, and we sketch a proof below.

\begin{theorem}[Theorem 1.2, Corollary 2.22  of~\cite{HolmKTZZ19}]
	Let $G$ be an arbitrary undirected $n$-vertex graph and let $k \ge c \log n$,
		where $c$ is a large enough constant. Let $G$ be a random
		subgraph of G where every vertex $v$ independent samples a subset of $\min(k,\deg(v))$ edges incident on it, each subset equally likely. 
		Then the expected number of edges in $G$ that connect different connected components of $G$ is
		$O(n/k)$. Furthermore, there exists a constant $b$ such that the probability that
		the number of edges in $G$ that connect different connected
		components of $G$ exceeds  $\ell\cdot bn/k$ is at most $2^{-\ell}$.
\end{theorem}

Fix a vertex $v$. Let $d$ be its degree. Let $k = c\log n$ with $c$ as in the above theorem. 
Note that if $d \leq 2c\log n = 2k$, then by a coupon collector argument, whp, our with repetition experiment will sample all the edges. We may assume, therefore, $d > 2k$, 
then again a coupon collector style argument shows that whp we will obtain at least $k$ distinct edges. Furthermore, by symmetry, every subset of $k$ distinct edges are going to be equally likely.
Therefore, we can apply the above theorem which implies the lemma (set $\ell = \log^2 n$). 

For the interested reader, we include a self contained proof (a weaker version of \cite{HolmKTZZ19} theorem) of the lemma in \Cref{sec:app-proof}.
\end{proof}

\subsection{Linear Queries}\label{sec:lin-graph-conn}
In this section we prove the following theorem.
\begin{theorem}\label{thm:rand-graph-conn-cross-formal}
	There exists a $1$-round (non-adaptive) randomized algorithm $\RandGraphConnCross(G)$ which makes at most $O(n\log^4 n)$-\cross-queries 
	and returns a spanning forest of $G$ with high probability. 
\end{theorem}

\noindent
Our main tool is a data structure, which we call  \MultiCutSampler,  which takes input a single partition $\Pi := (S_1, \ldots, S_p)$ of the vertex set, and, with high probability, for each $i\in [p]$ returns an 
edge uniformly at random from $E(S_i,S^c_i)$.
The key feature of \MultiCutSampler is that it can be constructed \emph{non-adaptively}. That is, we next show how to create the \MultiCutSampler by making $O(n\log^3 n)$ \linear (indeed, \ADD) queries to the graph.

%

\begin{remark}\label{rem:mcs-once}
	It is important to assert that the above guarantee provided by \MultiCutSampler in response to a given query holds for answering \emph{exactly one query partition}. In other words, asking multiple queries from \MultiCutSampler \emph{adaptively} would ensure that the
	returned solution is wrong with probability close to one. Therefore, if one desires to find random edges for $k$ different partitions, it is necessary to construct $k$ independent copies of \MultiCutSampler.
\end{remark}


\paragraph{Notation.} For any set $S \subseteq V$, we use $\partial(S)$ to denote the set of pairs in the cut $(S,V\setminus S)$ which have edges between then.
We let $w(\partial(S))$ denote the number of these edges.
If $S = \set{v}$ for some $v \in V$, 
we slightly abuse the notation and use $\partial(v)$ instead of $\partial(\set{v})$. Additionally, we use $\partial^+(v)$ (resp. $\partial^-(v)$) to denote the set of edges that are go to lexicographically larger (resp. smaller) vertices in $\partial(v)$. 

\paragraph{Preprocessing.} 
Let $L := \set{2^{i} \mid 0 \leq i \leq \log{{{n}\choose{2}}}}$ 
and $k = 100\log{n}$. For any $\ell \in L$, sample $k$ sets $T^{\ell}_1,\ldots,T^{\ell}_k$, each chosen by picking each of the ${{n}\choose{2}}$ vertex-pairs independently with probability $1/\ell$. For any set $T^{\ell}_i$ for $\ell \in L$ and $i \in [k]$, 
sample $k$ random subsets of $T^{\ell}_i$ denoted by $T^{\ell}_{i,1},\ldots,T^{\ell}_{i,k}$ each chosen by picking each element in $T^{\ell}_i$ w.p. $1/2$ independent of other elements. This step can be done without any interaction with the graph $G$. 

\paragraph{Initialization.} For any vertex $v \in V$, any set $T^{\ell}_{i,j}$ for $\ell \in L$, and $i,j \in [k]$, let: 
\begin{align*}
t^{\ell}_i(v) &= w\left(\partial^+(v) \cap T^{\ell}_{i}\right) - w\left(\partial^{-}(v) \cap T^{\ell}_{i}\right), \\
t^{\ell}_{i,j}(v) &= w\left(\partial^+(v) \cap T^{\ell}_{i,j}\right) - w\left(\partial^{-}(v) \cap T^{\ell}_{i,j}\right).
\end{align*}

It is easily verified that for any vertex $v$, the scalar $t^{\ell}_{i,j}(v)$ can be computed using $2$-\cross-queries. For each vertex $v$, one recognizes the set $W$ of vertices $u\in V\setminus v$ such that the pair $(u,v)\in T^{\ell}_{i,j}$; one then queries $\cross(\{v\}, W)$ to get 
$w\left(\partial^+(v) \cap T^{\ell}_{i,j}\right)$.
The \MultiCutSampler algorithm simply calculates $t^{\ell}_{i,j}(v)$ for all vertices $v$ and all choice of parameters $\ell,i,j$ in the initialization step. This takes $O(n\log^3 n)$-\cross queries.

\paragraph{Query Response.} Given a partition $\Pi$, and for all sets $S\in \Pi$, 
\begin{enumerate}
	\item For all $\ell\in L,i\in [k],j\in [k]$, compute $t^{\ell}_{i}(S) := \sum_{v \in S} t^{\ell}_{i}(v)$ and $t^{\ell}_{i,j}(S) := \sum_{v \in S} t^{\ell}_{i,j}(v)$.
	\item \emph{Level test:} Find $(\ell^*,i^*)$ such that 
	for all $j \in [k]$, $t^{\ell^*}_{i^*,j}(S) \in \set{0,t^{\ell^*}_{i^*}(S)}$. If no such choice exists, return $\perp$ and terminate. 
	\item \emph{Edge test:} Find a pair $e \in T^{\ell^{*}}_{i^*}$ such that for all $j \in [k]$, whenever $e \in T^{\ell^*}_{i^*,j}$, $t^{\ell^*}_{i^*,j}(S) = t^{\ell^*}_{i^*}(S)$. If no such edge exists, return $\perp$ and terminate. 
	\item Return the pair $e$ found in the last test, and assert this is in $E(S,S^c)$. 
\end{enumerate} 
\noindent
In the following, we prove the correctness of \MultiCutSampler. 

\begin{lemma}[Level Test Correctness]\label{lem:level-test}
	For parameters $(\ell^*,i^*)$ chosen by level test, $\card{T^{\ell^*}_{i^*} \cap \partial(S)}=1$ with high probability.
\end{lemma}
\begin{proof}
	We first prove that the probability that the level test outputs $\perp$ is at most $1/n^{O(1)}$. Let $\ell \in L$ be such that $\ell \leq \card{\partial(S)} < 2\ell$. For this choice of $\ell$ and by randomness in choice of $T^{\ell}_i$ for $i \in [k]$, 
	\begin{align*}
	\Pr\paren{\card{T^{\ell}_{i} \cap \partial(S)} = 1} = \Omega(1). 
	\end{align*}
	As a result, with high probability, there exists some $i \in [k]$, such that $\card{T^{\ell}_{i} \cap \partial(S)} = 1$. It can be verified that such choice of $(\ell,i)$ pass the level test because there is only one edge in $\partial(S)$ that can contribute
	to the value of $t^{\ell}_i(S)$ (the other edges will cancel out each other's contribution). Clearly, such a choice for $(\ell,i)$ satisfies the requirement in the lemma statement. We now prove that any $(\ell,i)$ such that
	$\card{T^{\ell}_{i} \cap \partial(S)} \neq 1$, would pass the test only with probability $1/n^{O(1)}$; this, together with a union bound on $\polylog{(n)}$ choices for $(\ell,i)$ finalizes the proof. 
	
	If $\card{T^{\ell}_{i} \cap \partial(S)} = 0$, clearly $(\ell,i)$ cannot pass the first step of the level test. Hence, in the following, we assume that $\card{T^{\ell}_{i} \cap \partial(S)} \geq 2$. 
	For any $j \in [k]$, we prove that, 
	\begin{align*}
	\Pr\paren{t^{\ell}_{i,j}(S) \notin \set{0,t^{\ell}_i(S)}} = \Omega(1),
	\end{align*}
	which immediately finalizes the proof as we repeat this process $k = O(\log{n})$ times. 
	
	Let $e_1$ and $e_2$ be two edges in $T^{\ell}_i \cap \partial(S)$. Suppose we fix the assignment of every edge in $\partial(S)$ in $T^{\ell}_{i,j}$ except for $e_1$ and $e_2$. There is always one choice of $e_1$ and $e_2$ that ensures
	that $t^{\ell}_{i,j}(S) \notin \set{0,t^{\ell}_i(S)}$. Hence, w.p. at least $1/4$, $t^{\ell}_{i,j}(S) \notin \set{0,t^{\ell}_i(S)}$, finalizing the proof. 
\end{proof}

\begin{lemma}[Edge Test Correctness]\label{lem:edge-test}
	Conditioned on $\card{T^{\ell^*}_{i^*} \cap \partial(S)}=1$, the edge $e$ returned by the edge test is the single edge in $T^{\ell^*}_{i^*} \cap \partial(S)$ with high probability. 
\end{lemma}
\begin{proof}
	Fix $e^{*} = T^{\ell^*}_{i^*} \cap \partial(S)$. It is immediate that $t^{\ell^*}_{i^*,j} = t^{\ell^*}_{i^*}$ iff $e^* \in T^{\ell^*}_{i^*,j}$. Now fix any other edge $e \neq e^* \in T^{\ell^*}_{i^*}$. The probability that $e$ appears in all 
	$T^{\ell^*}_{i^*,j}$ in which $e^*$ also appears is at most $1/2^{k} = 1/n^{O(1)}$. A union bound on all possible edges finalizes the proof. 
\end{proof}

The correctness of \MultiCutSampler now follows immediately from~\Cref{lem:level-test} and~\Cref{lem:edge-test}.

\subsubsection{The Connectivity Algorithm}\label{sec:uconn-algorithm}

The algorithms is simply as follows: 
	\begin{enumerate}
		\item Create $t:= O(\log{n})$ \MultiCutSampler data structure $D_1,\ldots,D_t$ by querying the graph non-adaptively. 
		\item Define $S^0_1 = \set{v_1},\ldots,S^0_{n}:=\set{v_n}$. 
		\item For $i = 1$ to $t$ steps: 
		\begin{enumerate}
			\item Query $D_i$ with sets $(S^{i-1}_1,\ldots,S^{i-1}_{n'})$ to obtain an edge from each $\partial(S^{i-1}_{j})$. 
			\item Let $S^{i}_j := S^{i-1}_{j} \cup S^{i-1}_{j'}$ whenever the edge sampled for $S^{i-1}_{j}$ by $D_i$ is incident on a vertex in $S^{i-1}_{j'}$.  
		\end{enumerate}
	\end{enumerate}
Since the algorithm finds an edge out of every connected cluster, and the number of clusters drop by a factor $2$ in each round, in $O(\log n)$-rounds the algorithm finds a spanning forest of $G$, whp. Once again, the total number of bad events is $O(n)$, and thus the whp holds due to a union bound over the whp assertions in~\Cref{lem:level-test} and~\Cref{lem:edge-test}. The total number of \cross-queries is in the creation of the $O(\log n)$ \MultiCutSampler data structures. This completes the proof of~\Cref{thm:rand-graph-conn-cross-formal}.

%% file: related.tex
\section{Related Work}\label{sec:related}

\paragraph{Graph Reconstruction via \cross Queries.} 
As mentioned in the Introduction, most of the work in the literature on \cross-queries in graphs has focused on reconstructing the graph.
Starting with the work of Grebinski and Kucherov~\cite{GrebinskiK98}, a long line~\cite{AlonBKRS02,AlonA05,ReyzinS07,ChoiK08,Bshouty09,BshoutyM11,Mazzawi10,BshoutyM12} of work culminated in a randomized, adaptive, {\em polynomial time} $O(\frac{m\log n}{\log m})$-\cross query algorithm due to Choi~\cite{Choi13} to reconstruct the graph. Interestingly, there is a non-adaptive, deterministic algorithm with the same number of queries~\cite{ChoiK08,Bshouty09}, however neither are the queries explicit, nor is there an efficient algorithm known to reconstruct the graph from the answers. To our knowledge, obtaining a polynomial time, non-adaptive algorithm with optimal query complexity is an open problem. Having said that, it is rather straightforward (see~\cite{ReyzinS07}, for instance) to obtain an efficient $O(m\log n)$-\cross query non-adaptive, deterministic algorithm.

One recent work using \cross-queries which is similar in spirit to our paper, is one by Rubinstein, Schramm, and Weinberg~\cite{RubinsteinSW18}. They give a randomized, $O(1)$-round\footnote{They don't specify the number of rounds, but our guess is $3$} algorithm which makes $\tilde{O}(n)$-\cross queries in a {\em simple} undirected graph and returns the global minimum cut. Their result, although related, is rather incomparable. For one, it uses more rounds than our algorithm, two, it runs on simple graphs, and lastly, it is unclear whether their algorithm can return a spanning tree. Our algorithm, on the other hand, cannot find the minimum cut. Indeed, for general multigraphs, a recent result of Assadi, Chen, and Khanna~\cite{AssadiCK19} shows that
any $O(1)$-round algorithm which finds the {\em exact} minimum cut must make $\Omgt(n^2)$ queries.

\paragraph{Parameter Estimation via \IS and \BIS Queries.} 
Beame et al.~\cite{BeameHRRS18} considered the problem of estimating the number of edges in a graph given $\BIS$-queries. They gave a randomized algorithm which obtained an $(1+\eps)$-approximation making $\poly\left(\log n, \frac{1}{\eps}\right)$-\BIS queries. 
This improved upon a earlier result of Dell and Lapinskas~\cite{DellL18} who gave a $\tilde{O}(n)$-query algorithm. 

Building upon~\cite{BeameHRRS18}, very recently Bhattacharya et al.~\cite{BhattacharyaBGM18} give a polylogarithmic query approximation to estimate the number of triangles (using a stronger query model). 
To our knowledge, we don't know of an explicit reference to graph reconstruction using these queries; however, it is not hard to obtain a
$2$-round, randomized algorithm making $O(m\log n)$-\BIS queries (assuming we know $m$); indeed, our algorithm in~\Cref{sec:algo-or} does that.

Apart from $\BIS$-query, another similar query model is the $\IS$ model which takes input a subset $A$ and says whether there is any edge with both endpoints in $A$. This has also been called the {\em edge-detecting model}~\cite{AngluinC08,AbasiB18}. We should stress here that the \IS queries are significantly weaker than \BIS queries. On the one hand, one can simulate $\IS$ queries using $O(\log n)$ nonadaptive \BIS queries.
On the other hand, 
it is known, for instance, that $O(\log n)$-\BIS queries can estimate the degree (\Cref{lem:rand-bis-degree-estimate}), but $\Omgt(n)$-\IS queries are needed to estimate the same~\cite{BeameHRRS18}. For our problem at hand (of finding a spanning forest), one can use the construction in~\Cref{sec:lb-rand-one} to show any (randomized, adaptive) algorithm for finding a spanning forest needs to make $\Omgt(n^2)$-\IS queries.

\paragraph{Sketching, Streaming, and other Access Models.} \linear queries are more famous as \linear sketches. Over the past two decases, a huge amount of literature has amassed on linear sketching; we refer the readers to surveys~\cite{WoodruffSurvey, McGregorSurvey} and the references within. Here, we mention the works most relevant to our paper.

Ahn, Guha, and McGregor~\cite{AhnGM12} were the first to give a \linear sketch for connectivity. As mentioned in the overview, their work implies a non-adaptive, randomized $\Ot(n)$-\linear query algorithm for finding a spanning forest. Our result (\Cref{Thm:rand-graph-conn-cross}) can be thought of as a special class of linear sketch for connectivity. One of the main applications of linear sketches arises in dynamic streaming algorithms. In a dynamic stream, objects are inserted and deleted in a stream, and the algorithm has to maintain a certain structure in bounded space. The AGM result~\cite{AhnGM12} immediately implied a $\Ot(n)$-space randomized, one-pass dynamic stream algorithm. We are, in fact, unaware of any results on {\em deterministic} dynamic stream algorithms for maintaining a spanning forest. Our result (\Cref{Thm:det-graph-conn}) implies an $O(n^{1+1/r})$-space algorithm in $O(r)$ passes, which in turn implies a $O(\log n)$-pass semi-streaming ($\Ot(n)$-space) algorithm.

Recently, Nelson and Yu~\cite{NelsonY18} proved an $\Omega(n\log^3 n)$-lower bound on the space requirement for the single-pass dynamic spanning forest problem. This, for instance, proves a lower bound of $\Omega(n\log^2 n)$ on the number of $\cross$-queries required by any non-adaptive algorithm (the information theoretic lower bound is only $\Omega(n)$).

Finally, we mention that many other query-access models have been proposed in the literature, and this approach is instructive and important to understand the power and limitations of algorithms. We mention one such recent work by Sun et al.~\cite{SunWYZ19} which consider querying an unknown {\em matrix} via matrix-vector queries. The paper studies multiple objectives arising from linear algebra, statistics, and most relevant to us, graph problems. For instance, the matrix could be the adjacency or edge-incidence matrix of an unknown graph. This model is (way) stronger than even the \linear-query model as every query returns $\Ot(n)$-bits of information. 
For connectivity,~\cite{SunWYZ19} show that if the matrix is edge-incidence, then results~\cite{kapralov2017single} on dynamic spectral sparsifiers imply $O(\polylog(n))$ queries suffice, while with an adjacency matrix, $\Omgt(n)$-queries are required.

\paragraph{Combinatorial Group Testing, Compressed Sensing, and Coin Weighing.}
The single element recovery problem is closely related to all these three deep fields. In group testing, we are given \OR-query access to a vector as in \ProblemOne, but the objective is to recover the whole support.
This field started with the work of~\cite{Dorfman43} out of a very practical application in World War II (we refer to the book~\cite{DH00} and various references within for historical perspectives), but it has since had numerous applications in fields as diverse as DNA screening~\cite{NgoD00} to multiaccess communication (MAC) protocols~\cite{Wolf85}. It is known that if there are $d$ elements in the support, then $O(d^2\min\left(\log(N/d), \log^2_d N\right))$ non-adaptive, deterministic queries suffice~\cite{KautzS64,D01,HwangS87,PoratR08} and $\Omega(d^2 \log_d N)$-queries are needed for a non-adaptive algorithm.
Closing this gap is an outstanding open question in the combinatorial group testing community.
On the other hand, with adaptive algorithms can solve this problem with $O(d\log(N/d))$-queries, and this is necessary. We are, however, unaware of any work understanding the trade-off with rounds of adaptivity.
We refer the interested reader to the book~\cite{DH00} and surveys~\cite{NgoD00} and lecture notes~\cite{Ngo11}.

Compressed Sensing and Coin Weighing problems are closer to the single element recovery problem with \linear queries. 
In the coin weighing problem~\cite{SoderbergS63}, one is given $N$ coins out of a collection
of coins of two distinct weights $w_0$ and $w_1$, a spring scale (as opposed to a balance-scale), and the objective is to determine the weight of each coin with minimal number of weighings. This is same as given a $0,1$-vector in $N$ dimensions, recover it using \linear queries.
There is a slew of work (we simply point the reader to the references in~\cite{Bshouty09}) on these problems.
It is known that with no other assumption, $\frac{2n}{\log n}$ is the correct answer for non-adaptive algorithms~\cite{Moser70,Lindstrom71}.
When it is known that there are $\leq d$ instances of one coin, then the best non-adaptive algorithm makes $O(d\log n)$-queries~\cite{Lindstrom71} while the best non-adaptive algorithm makes $O(\frac{d\log N/d}{\log d})$-queries~\cite{Bshouty09}. When the vector is an arbitrary non-negative vector, then the problem falls in the realm of (non-negative) compressed sensing. One of the main problems in compressed sensing is given a $d$-sparse (or close to $d$-sparse) $N$-dimensional vector, can it be recovered (or approximately recovered) from linear measurements. It is now well known that $O(d\log(N/d))$ non-adaptive measurements suffice~\cite{Donoho06,CandesRT06,CormodeM06}, and this is tight. There are various nuanced results in what it means by close to sparse and approximation, and we point the reader to the survey~\cite{GilbertI10} and the references within for a deeper picture.

%% file: appendix.tex
\input{lb-single-element-or.tex}

\section{Self contained proof of Lemma~\ref{lem:rand-or-conn-comp}}\label{sec:app-proof}
\begin{lemma}
	Let $G$ be an undirected graph on $n$ vertices. Suppose every vertex $v\in G$ samples $O(\log^2 n)$ edges with repetition, and let
	$\Pi = (S_1, \ldots, S_p)$ be the resulting connected components. 
	Then, whp, $|E_\mathsf{cross}(\Pi)| = O(n\log n)$.
\end{lemma}
\begin{proof}
	We analyze the above by deferring the $O(\log^2 n)$ edges per vertex over $L = O(\log n)$ phases.
	Let $E^R$ be the collection of sampled edges initialized to $\emptyset$.
	We also maintain a set $\tilde{E}$ of edges in the graph initialized to $\emptyset$. These will be edges we will ``give up'' on.
	We use this simple fact.
	\begin{fact}\label{fact:trivial}
		Let $\Pi$ denote the connected components induced by $E_R$ and $\Pi'$ the connected components induced by $E^R \cup \tilde{E}$. Then, $E_\mathsf{cross}(\Pi) \subseteq E_\cross(\Pi') \cup \tilde{E}$. In particular, $|E_\mathsf{cross}(\Pi)|\leq |E_\mathsf{cross}(\Pi')| + |\tilde{E}|$.
	\end{fact}

	In each phase, we will either sample $O(\log n)$ edges on every vertex, or we will add $O(n)$ edges to $\tilde{E}$.
	%
	Let $E^R_t$ be the collection of sampled edges in the first $t-1$ phases.
	Let $\Pi_t$ be the partition of the vertices formed by the connected components of $E^R_t \cup \tilde{E}$. 
	\noindent
	We call a component $C$ in $\Pi_t$ {\em big} if $|E(C,C^c)| > |C|$, and \emph{small} otherwise. Big components have the following useful property.
	
	\begin{claim}\label{clm:big}
		Let $C$ be a big component. Now suppose every vertex $v\in C$ samples $O(\log n)$ edges incident to it independently with replacement.
		Then whp, an edge in $E(C,C^c)$ is sampled.
	\end{claim}
	\begin{proof}
		Let $\mathcal{E}$ be the event that we {\em don't} sample an edge from $E(C,C^c)$.
		Rename the vertices in $C$ as $v_1, \ldots, v_\ell$. Let $\alpha_i$ denote the number of edges $v_i$ has to $C^c$, that is, $\alpha_i := |E(v_i,C^c)|$. Thus, we have $\sum_{i=1}^\ell \alpha_i > \ell$.
		The probability that $v_i$ {\em doesn't} sample an edge in $E(v_i, C^c)$ in the $t$th phase is $\leq \left(1 - \frac{\alpha_i}{\alpha_i + \ell}\right)^{O(\log n)}$ since $v_i$ has at most $\alpha_i + \ell$ edges incident on it.
		Therefore, 
		$
		\Pr[\mathcal{E}] = \prod_{i=1}^\ell  \left(1 - \frac{\alpha_i}{\alpha_i + \ell}\right)^{O(\log n)}.
		$
If any of the $\alpha_i > \ell$, we get that one of the product terms in the RHS becomes $< \frac{1}{\poly(n)}$. 
	If all $\alpha_i \leq \ell$, then
		$
		\Pr[\mathcal{E}^{(t)}] = \prod_{i=1}^\ell  \left(1 - \frac{\alpha_i}{\alpha_i + \ell}\right)^{O(\log n)}$ which is at most $\prod_{i=1}^\ell \left(1 - \frac{\alpha_i}{2\ell}\right)^{O(\log n)} \leq e^{-\frac{O(\log n)}{\ell}\sum_{i=1}^\ell \alpha_i} \leq  1/\poly(n)$ since $\alpha_i$'s sum to $>\ell$.
	\end{proof}
	
	Let $n^{(t)} = |\Pi_t|$, and let $n^{(t)}_b$ and $n^{(t)}_s$ denote the number of big and small components, respectively.
	Note that, the total number of edges in $\tilde{E}^{(t)} := \cup_{C\in \Pi_t: \textrm{small}} |E(C,C^c)| = O(n)$. 
	Among the $n^{(t)}_b$ big components, some of these have edges to small components, and some of these don't.
	Let $n^{(t)}_{b,1}$ be the first number and $n^{(t)}_{b,2}$ be the second. If $n^{(t)}_{b,1} < n^{(t)}_{b,2}$, then, we {\em coarsen} $\Pi_t$ to $\Pi_{t+1}$ by adding all the edges of $\tilde{E}^{(t)}$ to $\tilde{E}$; it is as if we are committing to these edges being in the final $E_\mathsf{cross}$.	The number of connected components in $\Pi_{t+1}$ is $\leq n^{(t)}_{b,1} + \theta$ where $\theta$ is the collection of ``new'' components. Note each such new component must contain a big component among to type 2 and a small component. Thus,
	$\theta \leq \min(n^{(t)}_{b,2},n_s)$. Therefore, since $n^{(t)}_{b,1} \leq n^{(t)}_{b,2}$, we get that $n^{(t+1)} \leq \frac{2}{3}n^{(t)}$. 
	On the other hand, if $n^{(t)}_{b,1} \geq n^{(t)}_{b,2}$, then we sample $O(\log n)$-edges per vertex. By~\Cref{clm:big}, the number of big components will then drop to $n^{(t+1)}_b \leq \frac{1}{2} n^{(t)}_{b,1} + n^{(t)}_{b,2} \leq \frac{3}{4} n^{(t)}_b$.
	
	Therefore, in each phase either the number of components or the number of big components drop by a constant factor. Two small components can't merge to give a big component. Therefore, in $O(\log n)$-phases, we end up either with a connected graph, or with a graph with all small components. We end the process in either case; note that in both cases if $\Pi_L$ is the final partition, we have $E_\mathsf{cross}(\Pi_L) = O(n)$.
	Now, we use~\Cref{fact:trivial}. Since $\Pi_L$ is induced by $E^R\cup \tilde{E}$, we get that $|E_\mathsf{cross}(\Pi)| \le |E_\mathsf{cross}(\Pi_L)| + |\tilde{E}| = O(n\log n)$. This is because, the number of phases is $O(\log n)$ and in each phase $|\tilde{E}|$ gets $O(n)$ edges.
\end{proof}

%% file: lb-single-element-or.tex
\section{Simple Lower Bound for Single Element Recovery with \OR-queries}\label{sec:lb-sing-el-or}
We provide a simple proof for the single element recovery problem with \OR-queries.
Note that since \OR-queries are weaker than \linear-queries, \Cref{thm:lb-lin-single-element} already implies this, but this proof is arguably simpler (or rather, more direct, or to use the language from the proof of~\Cref{thm:lb-lin-single-element}, it works in the {\em primal} space).
\begin{theorem}\label{thm:lb-single-element-or}
	Any $r$-round deterministic algorithm for \ProblemOne must make $\geq \left(N^{1/r} - 1\right)$-\OR queries in some round.
\end{theorem}
\begin{proof}
	Suppose, for the sake of contradiction, there is an algorithm which makes $< N^{1/r} - 1$ queries every round.
	The proof is via an adversary argument. 
	Given an algorithm $\ALG$, the \adv maintains a vector $x\in \{0,\star\}^N$ and responds queries consistent with $x$; the $\star$ indicates that the \adv has not committed to the value of $x$ on that coordinate. 
	At the end of $r$-rounds, the algorithm returns a coordinate $j$. Our goal is to reveal $x\in \{0,1\}^N$ at that point such that $x_j = 0$ and $x$ is consistent with the transcript so far. If we are successful, then the algorithm cannot be correct on all inputs. We call this {\em fooling} the algorithm. 
	
	The \adv maintains a set of active vertices $A$. Initially $A = [N]$. For every $i\in A$, we have $x_i = \star$; for every $i\notin A$, we have $x_i = 0$. 
	Every query $Q$ made will be responded either $0$ or $1$. The former set are called $0$-queries, and \adv will maintain $x_i = 0$ for all $i$ in a $0$-query.
	All others are $1$-queries. The adversary maintains the following invariant.
	\begin{itemize}
		\item[(I.)] After round $k\geq 1$, for every $1$-query $Q$ with $Q\cap A \neq \emptyset$, we have $|Q\cap A|> N^{1 - k/r}$.
	\end{itemize}
Note that at the beginning, that is after round $k=0$, the above invariant holds vacuously.
%
%
%
	
	Consider an arbitrary round $1 \le k \leq r$. Let $Q_1, Q_2, \ldots, Q_t$ be the queries in this round, with $t < N^{1/r} - 1$.
	Call $Q_r$ {\em small} if $|Q_r \cap A| \leq N^{1-k/r}$. The \adv responds $0$ to every small query and (a) sets $x_i = 0$ for $i\in Q_r\cap A$, and (b) removes these elements from $A$.
	Note that each query can lead to a drop of $\leq N^{1-k/r}$ in the size of $|A|$.
	Since $A$ has diminished in size, this may make some other query small, and the \adv repeats this process till no small queries remain. For all such queries, the \adv responds $1$.
	This completes the description of the responses, and now let us show that the invariant holds. Indeed, fix any $1$-query $Q$. If this query $Q$ is from round $k$, then the invariant holds by the description
	of the \adv process: otherwise, $Q$ would be small and therefore a $0$-query. If $Q$ is from a previous round, then
	since the invariant held after round $(k-1)$, we know that before round $k$, $|Q\cap A| > N^{1 - (k-1)/r}$. After round $k$, the set $A$ decreases, but by at most $N^{1-k/r}\cdot \left(N^{1/r}-1\right)$
	since there are $< N^{1/r} - 1$ queries in all (perhaps all queries are small). Thus, after round $k$, we still have $|Q\cap A| > N^{1 - (k-1)/r} - N^{1-k/r}\cdot \left(N^{1/r}-1\right) = N^{1-\frac{k}{r}}$.
%
	
	To complete the proof of the theorem, note that at the end of $r$ rounds, we are in a situation where there is a subset of elements $A$ and every query $Q$ ever made by the algorithm is either $Q\cap A = \emptyset$ and we have responded $0$,  or $|Q\cap A| \geq 2$ and we have responded $1$. This is trouble for the algorithm. Suppose the algorithm returns $j\in [N]$ claiming that $x_j > 0$. Well, consider the vector $x$ which is $x_i = 0$ for all $j \cup [N]\setminus A$ and $x_i = 1$ otherwise. 
	We claim this is consistent with every query --- if $Q\cap A = \emptyset$, then we respond $0$ as we should, and if $|Q\cap A| \geq 2$, then we respond $1$. Since $|Q\cap A|\geq 2$, there must exist some element of $A\setminus j$ in $Q$. That is, $x$ has at least one endpoint in $Q$.
	Therefore, we are consistent. This proves the algorithm's behavior is incorrect completing the proof of this theorem.
\end{proof}